\documentclass[letterpaper,10pt,twocolumn]{IEEEtran}
\usepackage{url}
\usepackage{cite}
\usepackage{amsmath} 
\usepackage{amssymb}  
\usepackage{amsthm}
\usepackage{lipsum,adjustbox}
\usepackage{multirow,color}

\usepackage[shortlabels]{enumitem}
\allowdisplaybreaks
\DeclareMathOperator{\lie}{\mathcal{L}}
\DeclareMathOperator{\one}{\mathbf{1}}
\DeclareMathOperator{\zero}{\mathbf{0}}
\DeclareMathOperator{\Lap}{\mathsf{L}}
\DeclareMathOperator{\A}{\mathsf{A}}
\DeclareMathOperator{\D}{\mathsf{D}}
\DeclareMathOperator{\M}{\mathsf{M}}
\DeclareMathOperator{\N}{\mathsf{N}}
\DeclareMathOperator{\PP}{\mathsf{P}}
\DeclareMathOperator{\V}{\mathcal{V}}
\DeclareMathOperator{\normal}{\mathcal{N}}
\DeclareMathOperator{\E}{\mathcal{E}}
\DeclareMathOperator{\F}{\mathcal{F}}
\DeclareMathOperator{\R}{\mathcal{R}}
\DeclareMathOperator{\im}{\operatorname{i}}
\DeclareMathOperator{\erf}{\operatorname{erf}}
\DeclareMathOperator{\rank}{\operatorname{rank}}
\DeclareMathOperator{\G}{\mathcal{G}}
\DeclareMathOperator{\refs}{\operatorname{ref}}
\newcommand{\real}{\mathbb{R}}
\newcommand{\integer}{\mathbb{Z}}
\newcommand{\gdac}{\psi_{\text{gdac}}}
\newcommand{\complex}{\mathbb{C}}
\renewcommand{\baselinestretch}{0.98}

\newtheorem{theorem}{Theorem}[section]
\newtheorem{proposition}[theorem]{Proposition}
\newtheorem{lemma}[theorem]{Lemma}

\theoremstyle{remark}
\newtheorem{remark}{Remark}
\theoremstyle{definition}

\newcommand{\longthmtitle}[1]{\mbox{} \textit{(#1):}}
\newcommand{\setdef}[2]{\{#1 \; | \; #2\}}
\newcommand{\map}[3]{#1:#2\rightarrow #3}

\graphicspath{ {Figures/} }
\begin{document}
\title{Enabling DER Participation in Frequency Regulation Markets}
\author{Priyank Srivastava \quad Chin-Yao Chang \quad Jorge
  Cort\'{e}s
  \thanks{A preliminary version of this work appeared at the American
    Control Conference as~\cite{PS-CYC-JC:18-acc}.}  \thanks{This work
    was supported by the ARPA-e Network Optimized Distributed Energy
    Systems (NODES) program, DE-AR0000695.}
  \thanks{P. Srivastava and J. Cort\'{e}s are with the Department of Mechanical and
    Aerospace Engineering, University of California, San Diego,
    {\tt\small \{psrivast,cortes\}@ucsd.edu}}
    \thanks
    {C.-Y. Chang is with the National Renewable Energy Laboratory,
    {\tt\small \{chinyao.chang\}@nrel.gov}}
}
\maketitle

\begin{abstract}
  Distributed energy resources (DERs) are playing an increasing role
  in ancillary services for the bulk grid, particularly in frequency
  regulation.  In this paper, we propose a framework for collections
  of DERs, combined to form microgrids and controlled by aggregators,
  to participate in frequency regulation markets.  Our approach covers
  both the identification of bids for the market clearing stage and
  the mechanisms for the real-time allocation of the regulation
  signal.  The proposed framework is hierarchical, consisting of a top
  layer and a bottom layer. The top layer consists of the aggregators
  communicating in a distributed fashion to optimally disaggregate the
  regulation signal requested by the system operator. The bottom layer
  consists of the DERs inside each microgrid whose power levels are
  adjusted so that the tie line power matches the output of the
  corresponding aggregator in the top layer. The coordination at the
  top layer requires the knowledge of  cost functions, ramp rates
  and capacity bounds of the aggregators. We develop meaningful
  abstractions for these quantities respecting the power flow
  constraints and taking into account the load uncertainties, and
  propose a provably correct distributed algorithm for optimal
  disaggregation of regulation signal amongst the microgrids.
 \end{abstract}

\section{Introduction}
Electric power systems require the generation and load to be equal at
all times. Any discrepancy between the two leads to the deviation of
the frequency from its nominal value.
This deviation of the frequency leads to many undesirable
scenarios. Based on measurements of the frequency deviation, the
system operator computes the automatic generation control (AGC) signal
as the feedback frequency control to the power system, which appears
as the total active power adjustment. Traditionally, frequency
regulation services have been provided by individual energy resources,
such as coal generation plants or gas turbines. Recently, there has
been a trend towards the integration of more DERs into the system to
provide these services while reducing thermal and CO$_2$ emissions.
Such integration leads to higher uncertainty in the bulk grid. At the
same time, as most DERs are inertialess, they can be effective for
frequency regulation due to their high ramp rates.  DERs are limited
in size and might not meet the minimum size criteria specified by
system operators to participate in the frequency regulation market. To
address these challenges, the vision is to integrate groups of DERs
through distributed energy resource providers (DERPs), or aggregators,
which would act as virtual power plants (VPPs) and would be
communicating with the system operator. These aggregators do not
necessarily own the DERs, they just coordinate their responses. This
architecture, illustrated in Figure~\ref{fig:fw}, has been proposed by
the California ISO (CAISO) to offer aggregators of DERs the
opportunity to sell into its marketplace~\cite{CAISO:15-derp}.  The
recent Order No. 2222~\cite{FERC-2222:20} by the U.S. Federal Energy
Regulatory Commission (FERC) also enables aggregators to participate
in the energy markets and requires all Regional Transmission
Organizations (RTOs) to revise their tariffs to establish DERs as a
category of market participant.  Using aggregators not only solves the
problem of limited capabilities of DERs but also enables the system
operator to interact with much fewer entities.  This paper is
motivated by the need to address the challenges to carry out the
vision described above.

\begin{figure}
\centering
    \includegraphics[width=0.47\textwidth]{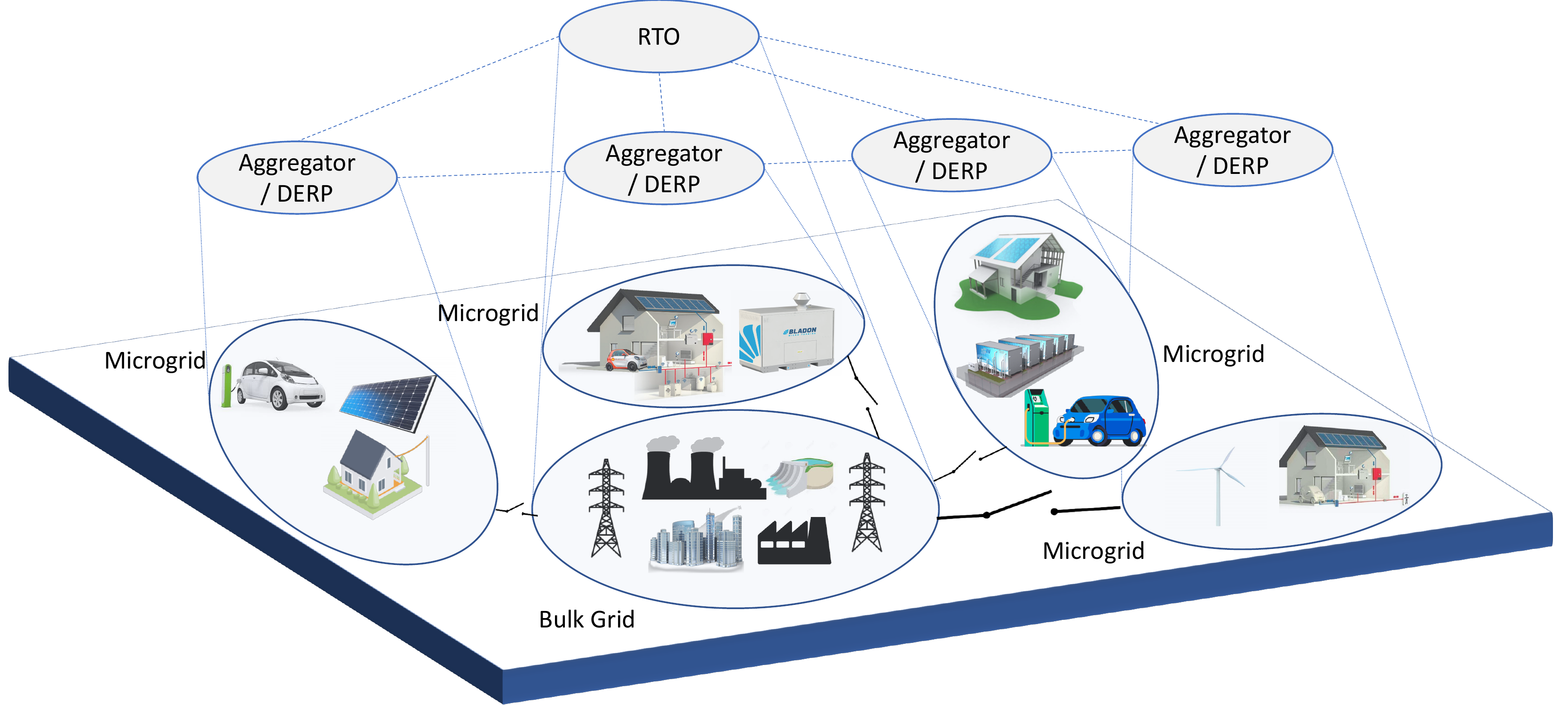}
    \caption{Power system framework. Dashed lines represent
      communication links and solid lines represent electrical
      connections. All the microgrids are connected to the bulk grid
      through the tie lines. The Regional Transmission Organization
      (RTO) monitors the bulk grid and coordinates with the
      aggregators, which control the resources inside the
      microgrids.}\label{fig:fw}
    \vspace*{-1.5ex}
\end{figure}

\emph{Literature Review:} Order No. 755~\cite{FERC:11} issued by the
FERC requires RTOs to compensate energy resources based on the actual
frequency regulation provided. The payment to resources comprises of
two parts, the capacity and performance payments. The capacity payment
compensates resources for their provision of regulation capacity.  The
performance payment reflects the accuracy of the tracking of the
allocated regulation signal. The work~\cite{MKM:14} describes how
different RTOs across the United States have implemented FERC Order
755 for participation of resources in frequency regulation market. In
the literature on power networks and smart grid, some works have
considered the possibility of obtaining frequency regulation services
from collections of homogeneous loads such as electric vehicles (EVs)
and thermostatically controlled loads (TCLs),
cf. \cite{JLM-SK-DSC:13,PC-MP-YP:14,BMS-HH-KP-TLV:14}. The
work~\cite{JTH-ADDG-KP:16} presents a method to model flexible loads
as a virtual battery for providing frequency
regulation.~\cite{SR-JS-HR:13} proposes the use of aggregators to
integrate heterogeneous loads such as heat pumps, supermarket
refrigerators and batteries present in industrial buildings to provide
frequency regulation. The works~\cite{OB-MP-YP:16,OB-KK-YP-MP-AP:18}
describe the challenges that need to be overcome for providing
frequency regulation by DERs for some European countries.  The
work~\cite{EDA-SG-AS-YCC-SVD:18} provides a framework to emulate
virtual power plants (VPPs) via aggregations of DERs and provide
regulation services taking into account the power flow
constraints.~\cite{BB-PA-TSP-KMN-JS-LHH:13} provides a dispatch
strategy for an aggregate of ON/OFF devices to provide frequency
regulation. In~\cite{JTH-ADDG-KP:17,CYC-SM-JC:17-cdc,HX-SCU-ADDG-PWS:17},
work has been done in the context of microgrids to design mechanisms
for optimally allocating a given signal among the DERs within the
microgrid.~\cite{RG-MA-PB:18} proposes a distributed algorithm to
minimize the aggregated cost while satisfying the local constraints
and collective demand constraint at the aggregator. However, the
aforementioned works assume that the allocated signal from the RTO is
available to the aggregator.~\cite{PM-AMK-HB-GD:16} applies machine
learning to forecast the power capacity of VPPs. The
work~\cite{YW-XA-ZT-LY-SL:16} provides a framework for optimal bidding
and dispatch of multiple VPPs.  \cite{SZ-YM-MS:17} proposes
  the use of renewable energy aggregators to utilize small-scale
  distributed generators for frequency regulation services via
  forecasting the available power from individual resources.  The
  work~\cite{SC-AM-GK:18} also uses forecasting to estimate the
  aggregate production from a wind and solar power-based VPP, and then
  uses the estimation to determine the optimal volume of reserves that
  can be provided to the system operator.  A distributed algorithm for
  coordinating multiple aggregators to provide frequency regulation,
  without any consideration of cost, is proposed
  in~\cite{JH-JC-JMG-TY-JY:17}.  Here, we focus on (i) participation
of microgrids in frequency regulation markets operated by the RTO
through the identification of appropriate bids and (ii) the
coordination among RTO and aggregators to efficiently dis-aggregate
the regulation signal amongst the aggregators. The actual tracking
performance within the microgrid would depend on the physical
condition of the resources.  We have provided some results for this
in~\cite{TA-MM-PS-HVH-JC-JK-SM-BW:20} on experiments carried out on
the University of California, San Diego (UCSD) microgrid.

\emph{Statement of Contributions:} We propose a hierarchical framework
for the participation of microgrids in the frequency regulation
market. We start by briefly reviewing the current practice of
frequency regulation from individual resources, consisting of three
stages: (i) market clearance, (ii) disaggregation of the regulation
signal and (iii) real-time tracking of the regulation signal.  Our
first contribution is the identification of the limitations of current
practice and the challenges that need to be overcome for integration
of microgrids. Our second contribution is the identification of
abstractions for the capacity, cost of generation, and ramp rates of a
microgrid as a combination of the individual energy resources that
compose it, along with a formal description of its convexity and
monotonicity properties. Building on our preliminary
work~\cite{PS-CYC-JC:18-acc}, here we extend our abstractions to the
case when the loads inside the microgrid do not remain constant for
the regulation period. Equipped with these abstractions, a microgrid
can submit bids to participate in the market clearance stage. Our
third contribution is the reformulation of the RTO-DERP coordination
problem to optimally disaggregate regulation signal amongst the
microgrids and accompanying design of an algorithmic solution. Our
proposed reformulation ensures feasibility. The proposed algorithm is
distributed over directed graphs with only one aggregator needing to
know the required regulation, and is guaranteed to asymptotic converge
to the desired optimizers.  We conclude with simulation results based
on the proposed abstractions of capacities, cost, and ramp rate and
the RTO-DERP coordination algorithm on a reduced-order model of the
University of California, San Diego (UCSD) microgrid.

\section{Preliminaries}\label{sec:prelims}
In this section, we present notational conventions and review some
basic concepts.

\emph{Notation:} Let $\complex$, $\real$, $\real_{\ge 0}$, and
$\integer$ be the set of complex, real, non-negative real and integer
numbers, respectively. For a set $|X|$, we let $|X|$ denote its
cardinality.  $\one$ and $\zero$ denote the vectors of all ones and
all zeros of appropriate dimension, respectively.  We use $|x|$ to
denote the absolute value of $x$, $[x]^+$ to denote $\max\{x,0\}$ and
$[x]^+_a$ to denote $[x]^+$ if $a>0$ and 0 if $a \leq 0$.  If $x$ is a
vector, these functions are applied elementwise.  For a matrix $A$,
its $i$th row and transpose are denoted by~$A_i$ and $A^\top$,
respectively.  We denote the gradient of a differentiable real-valued
function $f: \mathbb{R}^n \rightarrow \mathbb{R}$ by $\nabla f$.

\emph{Graph Theory:} We let $\G=(\V,\E,\A)$ denote a directed graph,
with $\V$ as the set of vertices (or nodes) and $\E\subseteq\V \times
\V$ as the set of edges.  $(i,j) \in \E$ iff there is an edge from
node $i$ to $j$.  We let $|\V|=n$ and $|\E|=m$.  A path is an ordered
sequence of vertices such that any pair of vertices that appear
consecutively is an edge.  A loop is a path in which the first and
last vertices are same and none of the other vertices is repeated.  A
graph is strongly connected if there is a path between any two
distinct vertices.  A tree is a graph whose underlying undirected
graph does not have any loops and is connected.  The \emph{adjacency
  matrix} $\A \in \real^{n \times n}$ of $\G$ is defined such that
$\A_{ij}>0$ if the edge $(i, j) \in \E$ and 0, otherwise.  The
out-degree and in-degree of a node $i$ are respectively, the number of
outgoing edges from and incoming edges to $i$.  The weighted
out-degree and the weighted in-degree of a node $i$ are given by
$d^{\text{out}}_i=\sum_{i=1}^{n} \A_{ij}$ and
$d^{\text{in}}_i=\sum_{i=1}^{n} \A_{ji}$, respectively.  The
\emph{weighted out-degree matrix} $\D^{\text{out}} \in \real^{n \times
  n}$ and the \emph{weighted in-degree matrix} $\D^{\text{in}} \in
\real^{n \times n}$ are the diagonal matrices with
$\D^{\text{out}}_{ii} = d^{\text{out}}_i$ and
$\D^{\text{in}}_{ii}=d^{\text{in}}_i$.  A graph is weight-balanced if
$\D^{\text{out}}=\D^{\text{in}}$.  The \emph{Laplacian matrix} $\Lap
\in \real^{n \times n}$ is defined as $\Lap = \D^{\text{out}} - \A$.
$0$ is a simple eigenvalue of $\Lap$ with eigenvector $\one$ iff $\G$
is strongly connected, and $\one^\top \Lap=\zero$ iff $\G$ is
weight-balanced.  The \emph{incidence matrix} $\M \in \real^{n \times
  m}$ is defined such that $\M_{ij} = 1$ if the edge $j$ leaves vertex
$i$, $-1$ if it enters vertex $i$, and $0$ otherwise.  Note that every
column of $\M$ has only two non-zero entries and $\one^\top \M=\zero$.
The \emph{fundamental loop matrix} $\N \in \real^{m \times (m-n+1)}$
of a graph has $\N_{ij}$ as 1 (-1, respectively) if the $i$th edge has
the same (opposite, respectively) orientation as the $j$th loop, and
$\N_{ij}=0$ if edge $i$ is not part of loop $j$. We use $\PP_{\refs}
\in \real^{(n-1) \times m}$ to denote the \emph{path matrix} of a tree
with reference vertex $\refs$: the $ij$th entry of the path matrix is
+1/-1 if edge $j$ is in the directed path from $i$ to $\refs$ and has
the same/opposite orientation as this path, and is 0 otherwise.
  
\emph{Probability Theory:} Given an event $E$, we let $E^c$ denote its
complement and $\Pr(E)$ its probability.  Given a normally distributed
random variable $\zeta \sim\normal(\mu, \sigma)$ with mean $\mu$ and
variance $\sigma$, the probability $\Pr(\zeta \leq x)$ of $\zeta$
being less than or equal to $x$ is denoted
\begin{align*}
  \Phi(x)= \int\limits_{-\infty}^{x} \dfrac{1}{\sqrt{2 \pi}}
  e^{-\frac{(\zeta-\mu)^2}{2\sigma}}d \zeta.
\end{align*}
The error function erf, defined as $ \erf(x)=\dfrac{2}{\sqrt{\pi}}
\int\limits_0^x e^{-\zeta^2} du$, denotes the probability of a normal
random variable with mean 0 and variance 1/2 being in the interval
$[-x,x]$. For a normal random variable with mean 0 and variance 1/2,
the functions $\Phi$ and erf are related by
\begin{align}\label{eq:erf}
  \Phi(x) = \frac{1}{2} \Big( 1 + \erf \big( \frac{x}{\sqrt{2}}
  \big)\Big).
\end{align}

\emph{Dynamic Average Consensus:} Consider a network of $n \in
\integer_{> 1}$ agents communicating over a strongly connected weight
balanced directed graph $\G$.  Each agent has a state $z_{i} \in
\real$ and an input signal $\map{u_{i}}{\real}{\real}$. Dynamic
average consensus aims at making each agent track the average input
$\frac{1}{n}\sum_{i=1}^n u_{i}(t)$ asymptotically. Formally, we employ
the dynamics given by
\begin{align*}
  \dot{z}&=\dot{u}- \nu (z - u) - \beta L z - v,
  \\
  \dot{v}&= \nu \beta \Lap z,
\end{align*}
where $\Lap \in \real^{n \times n}$ is the Laplacian of $\G$ and $\nu,
\beta> 0$ are the design parameters.  If the algorithm is initialized
with $\one^\top v(0)=0$, then the steady-state error between the state
$z_i$ of each agent $i \in \{1, \ldots, n\}$ and the average signal
$\frac{1}{n}\sum_{i=1}^n u_i$ is bounded, and goes to zero if $\dot{u}
\to \zero$, cf.~\cite[Theorem 4.1]{SSK-JC-SM:15-ijrnc}.

\section{Frequency Regulation with Microgrids}\label{current}
We are interested in coordinating power aggregators to collectively
provide frequency regulation.  An aggregator is a virtual entity that
aggregates the actions of a group of distributed energy resources to
act as a single whole. In this paper, we identify an aggregator with a
microgrid, but in general it may correspond to other entities (such
as, for instance, a collection of microgrids).

\subsection{Review of Current Practice}\label{sec:review}
The frequency regulation market is operated by an RTO to regulate the
system frequency at its nominal value. To achieve this, the RTO
coordinates the response of participating energy resources in a
centralized fashion to assign the regulation signal and restore the
power balance of the grid. Different RTOs follow slightly different
procedures for the frequency regulation markets. The procedure
followed by CAISO has the following stages, see e.g.,~\cite{MKM:14}:
\begin{figure}
\centering
    \includegraphics[width=0.47\textwidth]{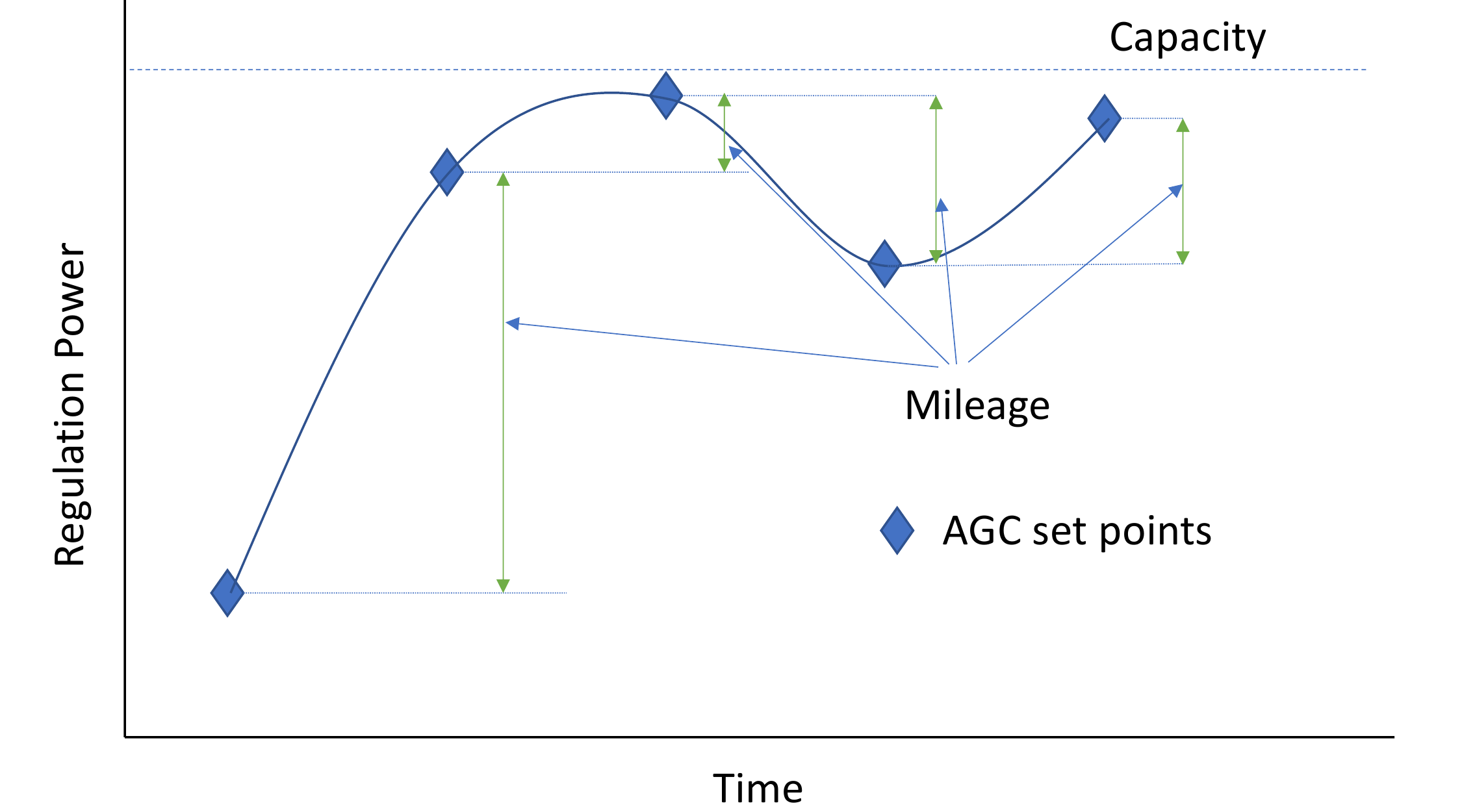}
    \caption{Illustration of the computation of capacity and
      mileage}\label{fig:capacity-mileage}
    \vspace*{-1ex}
\end{figure}

\textbf{[CP1]: Market clearance.} All participating resources submit
their capacity bids, capacity price bids, and mileage price bids to
the RTO. Capacity bids are the maximum amount of regulation (up or
down) that the resource can provide. Capacity price bids are the unit
price of providing these regulations. Mileage is the sum of the
absolute change in AGC set points, which corresponds to the summation
of the vertical lines in Figure~\ref{fig:capacity-mileage}. The
mileage price bid is the cost for unit change in
regulation. Typically, expected mileages are calculated from
historical data and resources do not submit mileage bids. The RTO
clears the market with a capacity and mileage price that is uniform
across the resources, and sends each resource its capacity and mileage
allocation. This off-line process happens only once per regulation
event.

\textbf{[CP2]: Allocation of regulation signal to each resource.}  The
RTO sends the regulation set points to each of the procured energy
resources every 2-4 seconds for the entire regulation period, which is
usually 10-15 minutes. The regulation set points are computed from the
AGC signal in real time in proportion to the procured mileage of each
resource. In case the assigned capacity of a resource is violated, the
overshoot power is redistributed to the other resources in proportion
to their assigned mileages.

\textbf{[CP3]: Real-time tracking of regulation signal.}  Once the
regulation set points have been assigned, the resources need to track
them in real time.

Payment to the resources comprises of two components, capacity payment
and mileage payment. The capacity payment is done based on the
assigned capacity in [CP1] while the mileage payment is done based on
the actual mileage provided which reflects the performance of the
resources while tracking the assigned signal in~[CP3].

\subsubsection{Limitations of Current Practice}
The centralized way of assigning the set points to the resources in
[CP2] relies on the fixed number of resources with fixed generation
capacities procured in [CP1], which are available for the entire
regulation period. This is problematic in the context of aggregators,
as they are subject to variabilities and uncertainties associated with
the DERs inside them. Even if the DERs inside the microgrid
participating stay during the regulation period, the users inside the
microgrid can change their power consumption, which in turn leads to
changes in the effective regulation capacity.  Furthermore, in current
practice, there is no direct consideration of the operational costs of
the resources, which may result in suboptimal power
allocation. Instead, we argue that the assignment of the regulation
signal should be done, at each time step, in a way that optimizes the
aggregate cost functions of the resources and takes into account their
(possibly dynamic) operational limits. We refer to this approach as
the RTO-DERP coordination problem. This idea has also been pointed out
in the past by CAISO for traditional energy resources,
cf.~\cite{JB-SMH-BFH:12}. The lack of robustness and the information
sharing requirements of centralized schemes motivate the investigation
of distributed schemes to solve the RTO-DERP coordination problem.

\subsubsection{Challenges for Frequency Regulation from Microgrids}
Here we describe the challenges specific to microgrid participation in
frequency regulation markets.  First, note that solving the RTO-DERP
coordination problem with microgrids requires the identification, or
rather the abstraction, of aggregate cost functions and regulation
capacity bounds based on the cost functions and flexibilities of their
DERs.  Second, the determination of capacity bids requires taking into
account the uncertainties associated with the microgrids.  There is a
need to calculate bids for each regulation interval, as they might
need to considerably change from one interval to the next.  Even
within a regulation interval itself, the power level of the
uncontrollable nodes might vary significantly.
Third, the determination of mileage bids has to take into account the
dependency of ramp rates on the composition and participation of the
individual DERs.  The current method of calculating expected mileages
in [CP1] makes sense for conventional resources as their ramp rates
are fixed and historical data provides reliable accuracy. In the case
of microgrids, individual resources may be changing over time and ramp
rates might not remain constant as a result. Also, the performance of
participating resources for one regulation period to another might be
substantially different.

\subsection{Problem Statement}\label{prob}
Consider $N$ microgrids, each controlled by an aggregator.  To enable
microgrid participation in the frequency regulation market, we focus
on [CP1] and [CP2].  Based on the discussion in
Section~\ref{sec:review}, we formalize the following problems.

\textbf{[P1]: Meaningful abstractions for the microgrid.}  To enable
the submission of bids in [CP1], each aggregator needs to quantify the
maximum up/down regulation capacity that the microgrid can provide,
the unit cost of providing such regulation, and the ramp rate at which
the microgrid can change its power level. 
Our first goal is therefore to provide meaningful abstractions for
these objects, and cost functions and ramp rate functions of the
microgrids for [P2] below, a problem we tackle in Section~\ref{sec:abstractions}.

\textbf{[P2]: RTO-DERP distributed coordination.}  The RTO-DERP
coordination problem for computing the set points for each resource
advocated for [CP2] consists of an economic dispatch problem with ramp
rate constraints \emph{at every instant of the regulation interval}.
Formally, for $x_r$ regulation at a given time instant, we have
\begin{align}\label{eq:iso}
  & {\min_{x}} & & f(x) = \sum\limits_{\im=1}^N f_{\im}(x_{\im})
  \\
  & \text{s.t.} & & \sum_{\im =1}^{N}x_{\im}=x_r \notag
  \\
  &&& \underline{x_{\im}} \leq x_{\im} \leq \overline{x_{\im}} \quad
  \forall \im \displaybreak[0] \notag
  \\
  &&& |x_{\im} - x_{\im}^-| \leq R_{\im}(x_{\im}^-) \quad \forall \im,
  \notag
\end{align}
where $x \in \mathbb{R}^N$ is the vector of regulation power from the
microgrids, $f_{\im} (x_{\im})$ is the cost of $x_{\im}$ regulation
for microgrid $\im$, $\underline{x_{\im}}$ and $\overline{x_{\im}}$
are the lower and upper bounds of regulation for microgrid ${\im}$
which are bounded by the solutions of~[P1] and determined by~[CP1] for
a specific regulation period, $x_{\im}^-$ is the regulation that the
microgrid $\im$ was providing at the previous instant, and
$R_{\im}(x_{\im})$ is the ramp rate of the microgrid when it is
providing regulation $x_{\im}$. 
%
Because of the ramp constraints present in~\eqref{eq:iso}, this
problem might not be always feasible (since mileage requirements set
by the RTO while clearing the market in [CP1] capture the average
mileage required, and not the extreme cases). In such cases, we want
to minimize the error between the procured regulation and the required
one. We tackle these in Section~\ref{sec:rto-derp}.

\section{Microgrid Abstractions}\label{sec:abstractions}
Consider a microgrid with $n \in \mathbb{Z}_{>1}$ buses, described by
$\G_m=(\V,\E,\A)$. Without loss of generality, we assume that the
first bus is connected to the bulk grid through a tie line.  We
partition the remaining set of buses as $\V_g\cup\V_l$, where $\V_g$
is the set of controllable nodes and $\V_l$ is the set of
uncontrollable nodes (or loads). Let $n=|\V|$, $n_g=|\V_g|$,
$n_l=|\V_l|$ and $m=|\E|$. Following~\cite{DF-MZ-ADDG:18},
  we assume that the lines connecting various buses inside the
  microgrid are lossless and inductive.  In case the electrical lines
  inside the microgrid are lossy with sufficiently uniform resistance
  to reactance ratios, they could still be represented via a lossless
  model obtainable through a linear
  transformation~\cite{FD-JWSP-FB:16-tcns}.  Since the voltage
  dynamics governed by the voltage droop controllers operate at much
  faster scale than the secondary frequency
  regulation~\cite{LL-SD:14}, we assume the voltage magnitude of every
  bus to be approximately 1~p.u.  Further, we assume that the network
and inverter filter 
dynamics are fast enough so that we can model them as power injections
with no dynamics~\cite{OA-MA-IC-PWS-ADDG:17,QCZ-TH:12}.  We adopt the
convention that the value of the power injection is negative if the
device consumes power and vice versa. The power level of each
controllable node $p \in \V_g$ is denoted by $g_{p}$, with $g_{p}^0$
denoting the baseline generation/consumption. The power level of each
uncontrollable node $q \in \V_l$ is denoted by~$l_{q}$.  We denote the
incoming power through the tie line by $P$ and its baseline value by
$P^0$.  When the microgrid provides frequency regulation, the value of
the tie line power $P$ is
\begin{align*} 
  P = P^0 + x ,
\end{align*}
where $x$ is the allocated AGC signal. Note that since we model $P$ as
the incoming power from bulk grid, $x$ would be negative when the
microgrid is providing up regulation.
Following~\cite{MZ-ADDG:20}, the power injections for the microgrid are given by
\begin{equation}\label{eq:pf1}
  \begin{bmatrix}
    P & g^\top & -l^\top
  \end{bmatrix}^\top =\M B\text{sin}(\M^\top\theta),
\end{equation}
where $g \in \mathbb{R}^{n_g}$ and $l \in \mathbb{R}^{n_l}$ are
the vectors of controllable and uncontrollable nodes, resp., 
$\M \in \mathbb{R}^{n \times m}$ is the incidence matrix of the graph,
$B \in \mathbb{R}^{m \times
 m}$ is the diagonal matrix of absolute line susceptances and
$\theta \in \mathbb{R}^{n}$ is the vector of phase angles.
Additionally, there is a constraint that the maximum phase angle
difference between any two interconnected buses should be bounded by
$\Theta \in \real_{\ge 0}$, i.e.,
\begin{equation}\label{eq:angle_diff}
  |\theta_i - \theta_k| \leq \Theta.
\end{equation}
To avoid dealing with the nonlinearity in~\eqref{eq:pf1}, 
%
we assume $\G_m$ is a graph with non-overlapping loops and
rewrite the power flow equations as
\begin{subequations}\label{eq:pff}
  \begin{align}\label{eq:pffa}
    &
    \begin{bmatrix}
      P & g^\top & -l^\top
    \end{bmatrix}^\top =\M\omega,
    \\
    \label{eq:pfcf}
    & \: |\omega| \leq \overline{\omega}.
  \end{align}
\end{subequations}
Here $\omega \in \mathbb{R}^{m}$ is the vector of line flows and
$\overline{\omega} \in \mathbb{R}^{m}$ is the vector of maximum
permissible flows.  Since the columns of the fundamental loop matrix
form a basis for the null space of the incidence matrix,
cf.~\cite[Theorem 4-6]{SS-MBR:61}, we write~\eqref{eq:pff}~as
\begin{align}\label{eq:pseudoinv}
  \big| \M^+
  \begin{bmatrix}
    \mathbf{1}^\top l - \mathbf{1}^\top g
    \\
    g
    \\
    -l
  \end{bmatrix} + \N \gamma \big| \leq \overline{\omega},
\end{align}
where $\M^+$ denotes the Moore-Penrose pseudoinverse of $\M$, $\N \in
\real^{m \times (m-n+1)}$ is the fundamental loop matrix of $\G_m$, and
$\gamma \in \real^{m-n+1}$.

\subsection{Capacity Bounds}\label{sec:bounds}
The microgrid needs to solve an optimization problem to find the
maximum up (or down) regulation that it can provide.  For up
regulation, the power consumption of the microgrid is less than the
baseline power. Since the latter is constant for the
regulation period, computing the capacity is equivalent to minimizing
$P$ while satisfying the power flow constraints. If the power level of
uncontrollable nodes is constant for the entire regulation period,
then the problem reads as
\begin{equation}\label{eq:upreg}
  \begin{aligned}
    & \min_{g,\omega} & & P
    \\
    & \text{ s.t.}  && \begin{bmatrix} P & g^\top & -l^\top
    \end{bmatrix}^\top=\M\omega
    \\
    &&& \; \underline{g} \leq g \leq \overline{g} , \quad |\omega|
    \leq \overline{\omega},
  \end{aligned}
\end{equation}
where $\underline{g}$ and $\overline{g}$ are the vectors of minimum
and maximum possible power levels of controllable nodes,
respectively. If $\underline{P}$ denotes the solution
of~\eqref{eq:upreg}, then the maximum up regulation is $
\overline{x}=\underline{P}-P^0$.  The maximum down regulation
$\underline{x}$ can be obtained solving a similar maximization
problem. 

The formulation~\eqref{eq:upreg} assumes the power level of the
uncontrollable nodes remains constant, and therefore does not take
into account the varying nature of the loads.  In practice, this makes
sense for a specific regulation instant, and would rarely be the case
for the whole regulation period.  Instead, a more robust way of
calculating the capacity bounds that the aggregator can use in bidding
for the whole regulation period is to account for worst-case
scenarios, i.e., taking the expected maximum value for the
uncontrollable nodes while computing the maximum up
regulation. Although robust to variations in the uncontrollable nodes'
powers, this way of computing capacity bounds might be too
conservative and, in fact, might prohibit the microgrid from
participating in the regulation market at all.  As an alternative, we
propose a reformulation of problem~\eqref{eq:upreg} based on chance
constraints.  Using~\eqref{eq:pseudoinv}, we rewrite the optimization
problem~\eqref{eq:upreg} as
\begin{equation}\label{eq:upregcombine}
  \begin{aligned}
    & \min_{g, \gamma, t} & & t
    \\
    & \text{ s.t.}  && t \geq \mathbf{1}^\top l - \mathbf{1}^\top g
    \\
    &&& \big| \M^+
      \begin{bmatrix}
        \mathbf{1}^\top l - \mathbf{1}^\top g
        \\
        g
        \\
        -l
      \end{bmatrix} + \N \gamma \big| \leq \overline{\omega}
      \\
      &&& \underline{g} \leq g \leq \overline{g}.
  \end{aligned}
\end{equation}
Assume that a probability distribution describing the power levels of
uncontrollable nodes at any instant of the regulation period is
available.  To account for load variability, we instead consider the
following chance-constrained optimization
  \begin{align}\label{eq:upregchance}
    & \min_{g, \gamma, t}
    & & t  \notag \\
    & \text{ s.t.}  && \Pr (t \geq \mathbf{1}^\top l -
    \mathbf{1}^\top g) \geq 1- \epsilon'
    \\
    &&& \Pr \Big( \big| \M^+ 
    \begin{bmatrix}
      \mathbf{1}^\top l - \mathbf{1}^\top g
      \\
      g
      \\
      -l
    \end{bmatrix}
    + \N \gamma \big|_j \leq \overline{\omega}_j \Big) \geq 1-\epsilon \quad \forall j
    \notag \\
    &&& \underline{g} \leq g \leq \overline{g}.\notag
  \end{align}
where $\epsilon', \epsilon \in [0,1]$. In this formulation, each flow constraint can be
violated, with a probability no more than $\epsilon$. 

Since the regulation period lasts for only a short period of
  time (10-15 minutes), the variation in the loads would not be
  significant and it is reasonable to assume it could be approximately
  characterized by a normal distribution. The next result, whose
proof is in the Appendix, shows that the chance-constrained
optimization~\eqref{eq:upregchance} can be solved via a deterministic
linear program if the loads are normally distributed.

\begin{lemma}\longthmtitle{Capacity bounds for variable loads via 
    deterministic optimization}\label{lemma:chance to linear}
  Assume the loads are distributed normally with mean $\hat{l}$ and
  variance $V_l$.  Then, the solution of the deterministic linear
  program
  \begin{align}
    \label{eq:chance-det}
    & \min_{g, \gamma, t} & & t
    \\
    & \textnormal{ s.t.}  && \mathbf{1}^\top \hat{l} - \mathbf{1}^\top g - t
    \leq \sqrt{2} \erf^{-1}(2 \epsilon'-1) (\mathbf{1}^\top V_l
    \mathbf{1})^{1/2} \notag
    \\
    &&& |(\M_{1}\mathbf{1}^\top-\M_{3})\hat{l} +
    (\M_{2}-\M_{1}\mathbf{1}^\top)g + \N \gamma| \leq  \overline{\omega}^l
    \notag
    \\
    &&& \underline{g} \leq g \leq \overline{g}, \notag
  \end{align}
  where $\M^+ = [\M_1 \quad \M_2 \quad \M_3]$ with $\M_1 \in \real^{m}$,
  $\M_2 \in \real^{m \times n_g}$ and $\M_3 \in \real^{m \times
    n_l}$, $\overline{\omega}^l = \overline{\omega} + K$ and
  \begin{multline*}
    K_j = \sqrt{2} \erf^{-1} ( \epsilon-1) \cdot
    \\
    \cdot
    \big((\M_{1j}\mathbf{1}^\top-\M_{3j})V_l
    (\M_{1j}\mathbf{1}^\top-\M_{3j})^\top\big)^{1/2},
  \end{multline*}
  is a solution of problem~\eqref{eq:upregchance}.
\end{lemma}

\begin{figure}[htb]
  \centering
  \includegraphics[scale=0.45]{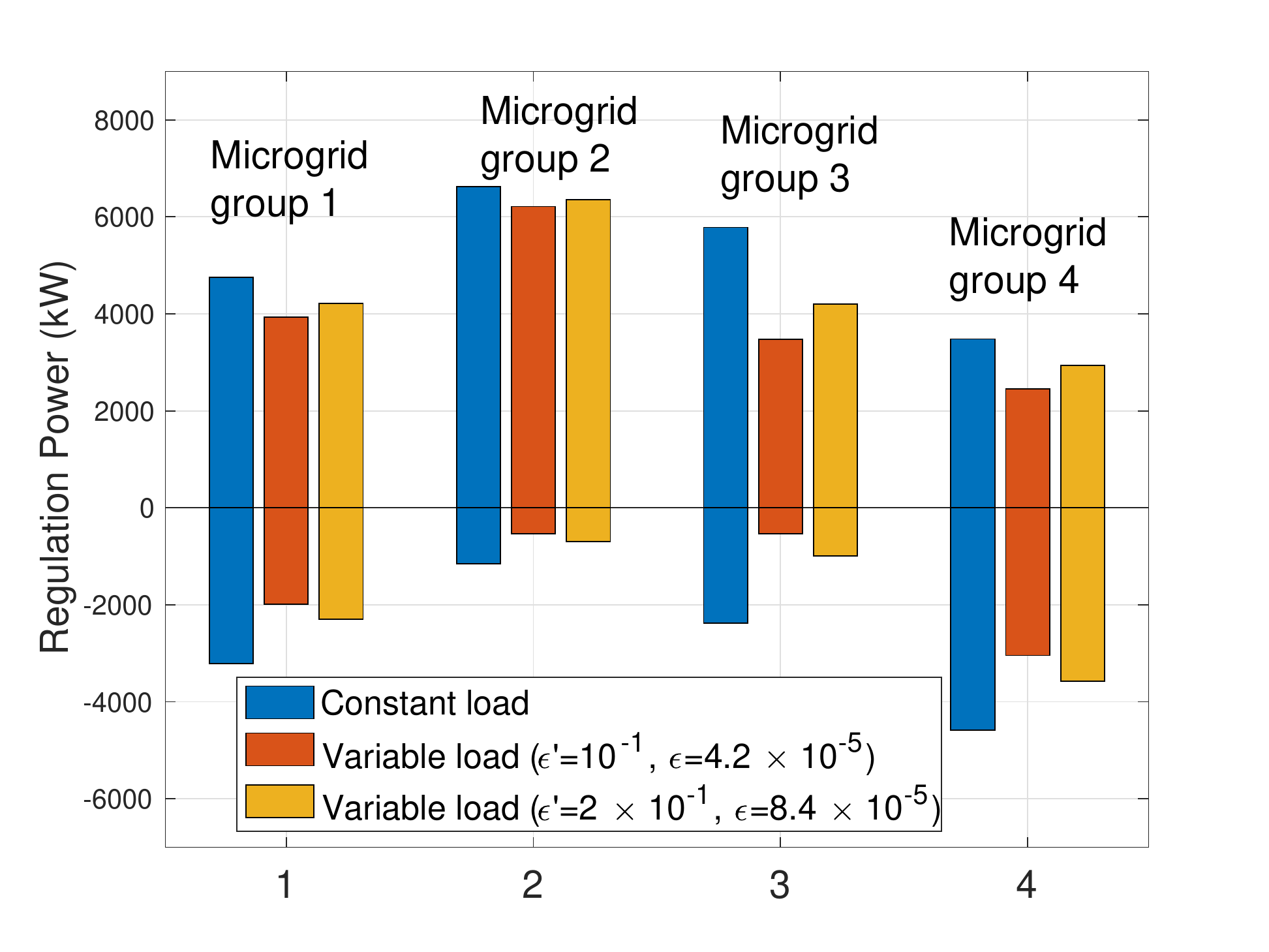}
  \caption{Regulation capacities for different instantiations of the
    reduced-order UCSD microgrid. The 12 microgrids are divided into 4
    groups, with constant mean value of loads and baseline generation
    across each group.  Within each group, for the first scenario, the
    variance is 0 and for the remaining two scenarios, it is given by
    the diagonal matrix which is 0.25 times the squared value of mean
    loads.  Values of $\epsilon', \epsilon$ for the second and third
    scenarios are $10^{-1}, 4.2 \times 10^{-5}$ and $2 \times 10^{-1},
    8.4 \times 10^{-5}$, respectively.  }\label{fig:capacity}
\end{figure}

We use Lemma~\ref{lemma:chance to linear} to compute in
Figure~\ref{fig:capacity} the maximum up and down regulation for
several microgrids modeled after the reduced-order UCSD microgrid
described later in Section~\ref{sec:sims}.  The microgrids are divided
into 4 groups, each with a different value of baseline generation and
mean load for the UCSD model. Within each group, we consider 3
different scenarios, one with constant load and the other two with
varying loads (generated using the same normal distribution with
variance equal to a diagonal matrix given by 0.25 times the squared
value of mean loads) and different confident values ($\epsilon',
\epsilon = 10^{-1}, 4.2 \times 10^{-5}$ and $2 \times 10^{-1}, 8.4
\times 10^{-5}$, respectively). One can see in
Figure~\ref{fig:capacity} that the capacity bounds increase
with~$\epsilon', \epsilon$, which is in agreement with the fact that
larger values of these correspond to lower probability of satisfying
the constraints.

Note that the probabilistic capacity bounds identified above and
obtained after solving~\eqref{eq:chance-det} are good only for the
bidding in [CP1].  The actual regulation bounds at a given regulation
instant still depend on the load at that instant.

\subsection{Ramp Rate Function}\label{sec:ramp}
In the following we discuss how to compute the ramp up rate for the
microgrid (the discussion for ramp down rate is analogous).  If there
were no constraints on the power flows, then the ramp rate of the
microgrid would be the summation of ramp rates of all the controllable
nodes.  However, the presence of flow constraints may prevent every
controllable node from ramping at its full capacity and as such, the
ramp rate is a function that depends on the operating point of the
controllable nodes.  Let $\F_g = \setdef{g \in \real^{n_g}}{\exists \;
  \omega \in \real^m \text{ satisfying~\eqref{eq:pff}}}$ denote the
set of feasible operating points for controllable nodes.  If the power
levels of the uncontrollable nodes are constant, then the ramp up
rate, $\R : \F_g \to \real_{\ge 0}$, is formally given by
\begin{align}\label{eq:rc}
  & \max_{\Delta g, \Delta \omega} & & \mathbf{1}^\top \Delta g
  \\
  & \text{ s.t.}  & & 
  \begin{bmatrix}
    (P-\mathbf{1}^\top \Delta g) & \hspace*{-1mm} (g+\Delta g)^\top &
    \hspace*{-2mm} -l^\top
  \end{bmatrix}^\top =\M(\omega+ \Delta \omega) \nonumber
  \\
  &&&\Delta g \leq r, \quad |\omega+\Delta \omega| \leq
  \overline{\omega} \nonumber,
\end{align}
where $r \in \real^{n_g}$ is the vector whose component $r_p$ is the
nominal ramping capacity of the controllable node~$p$, and
$\omega+\Delta \omega$ is the vector of line flows corresponding to
the operating point $g+ \Delta g$.  

If the power levels of the uncontrollable nodes are variable, we use
chance-constraints as in the case of capacity bounds and the ramp up
rate $\R$ is given by
\begin{align}\label{eq:rc_chance}
  & \max_{\Delta g, \gamma} & & \mathbf{1}^\top \Delta g
  \\
  & \textnormal{s.t.}  & & \hspace*{-17pt} \Pr \Big( \big| \M^+
  \!\! \begin{bmatrix} \one^\top l - \one^\top(g+ \Delta g) \\
    g+\Delta g \\ -l
  \end{bmatrix}\!\! + \N \gamma \big|_j \le \overline{\omega}_j \Big)
  \ge 1-\epsilon \; \forall j \nonumber
  \\
  &&& \hspace*{-13pt} \Delta g \leq r. \notag
\end{align}
The following result, whose proof is similar to that  of
Lemma~\ref{lemma:chance to linear} and omitted to avoid repetition,
converts the chance-constrained optimization~\eqref{eq:rc_chance} into
a deterministic linear program if the loads are normally distributed.

\begin{lemma}\longthmtitle{Ramp rate for variable loads via deterministic 
    optimization}\label{lemma:chance to linear ramp}
  Assume the loads are distributed normally with mean $\hat{l}$ and
  variance $V_l$.  Then, the solution of the deterministic linear
  program
  \begin{align}
    \label{eq:chance-det-rc}
    & \max_{\Delta g, \gamma} & & \one^\top \Delta g
    \\
    & \textnormal{s.t.}  && \hspace*{-5pt}
    |(\M_{1}\one^\top-\M_{3})\hat{l} + (\M_{2}-\M_{1}\one^\top)(g+
    \Delta g) + \N \gamma| \leq \overline{\omega}^l \notag
    \\
    &&& \hspace*{-3pt} \Delta g \leq r, \notag
  \end{align}
  where $\M_1$, $\M_2$, $\M_3$ and $\overline{\omega}^l$ are as
  defined in Lemma~\ref{lemma:chance to linear}, is a solution for
  problem~\eqref{eq:rc_chance}.
\end{lemma}

The next result states the properties of the ramp rate
function~\eqref{eq:rc} for a tree network.  The proof, given in the
Appendix, is based on the description of the feasible region in terms
of the power levels of the controllable nodes stated in
Lemma~\ref{lemma:simplify}.  For the ramp rate function with normally
distributed loads defined in~\eqref{eq:rc_chance}, one can obtain a
similar result following Lemma~\ref{lemma:chance to linear ramp} (with
$\overline{\omega}$ replaced by $\overline{\omega}^l$).

\begin{proposition}\longthmtitle{Ramp rate of tree
    network}\label{prop1}
  Let $\G_m$ be a tree and $H$ denote the hyperrectangle describing
  the region of operation of the controllable nodes, where opposite
  faces correspond to the minimum and maximum possible power level of
  a controllable node.  Then the ramp rate $\R$ is piecewise affine on
  $H$, i.e., for some $s>0$, $H$ admits a decomposition
  \begin{align*}
    H = V_1 \cup V_2 \cup \ldots \cup V_s,
  \end{align*}
  where $\{V_{\alpha}\}_{\alpha=1}^s$ are polyhedra, and $\R$ is
  affine on each $V_{\alpha}$.
\end{proposition}

\begin{remark}\longthmtitle{Ramp rate for networks with
    non-overlapping loops}
  If the network is not a tree, then the flows corresponding to a
  power injection vector are not unique. Nevertheless, the ramp rate
  for networks with non-overlapping loops is a non-increasing function
  of $g$, as the feasible region of \eqref{eq:rc} can only shrink with
  increase in some component(s) of~$g$.
  \hfill $\bullet$
\end{remark}

Given a regulation power $x$, we note that there may be more than one
feasible operating point for the microgrid that produces it. As a
result, the ramp rate as a function of regulation power is not
uniquely defined.
We address this by defining $R: [\overline{x},\underline{x}] \to \real_{\ge 0}$, as
\begin{align*}
  R(x)={\max\limits_{g^*}}\: \R(g^*) ,
\end{align*}
where $g^*$ denotes a minimizer of the cost of producing the
regulation~$x$ while respecting the power flow and capacity
constraints. We take the maximum, since the optimizer $g^*$ might not
be unique.  If the cost functions for all the controllable nodes are
convex, each $g^*$ is a decreasing function with respect to $x$, which
means that at least one component of $g^*$ would decrease as $x$
increases (using the convention that up regulation is negative). Using
this fact, we conclude that $R$ as a function of $x$ is
non-decreasing, with maximum possible value as~$\mathbf{1}^{\top}r$.
Figure~\ref{fig:ramp} provides the ramp rate functions of the four
groups of microgrids displayed in Figure~\ref{fig:capacity} in the
constant load case.

\begin{figure}[htb]
  \centering
  \includegraphics[scale=0.45]{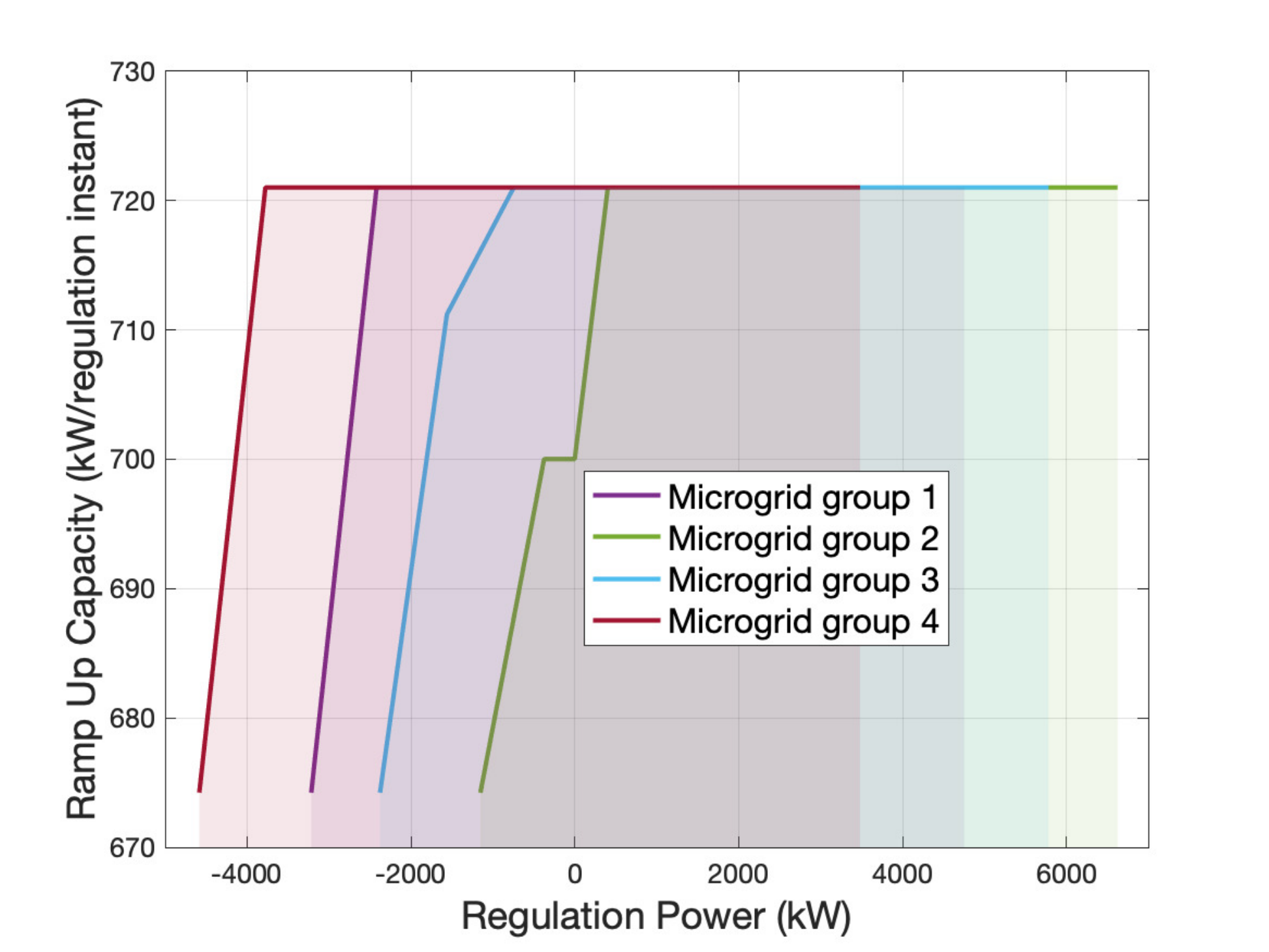}
  \caption{Ramp rate functions for different instantiations of the
    reduced-order UCSD microgrid with constant loads. The shaded
    regions represent the range of regulation power that the
    corresponding microgrid can provide.}\label{fig:ramp}
\end{figure}

In Remark~\ref{rem:ramp}, we discuss the
conditions under which the minimum ramp rate of the microgrid is
always non-zero.

\begin{remark}\longthmtitle{Non-zero minimum ramp rate}\label{rem:ramp}
   It is natural to argue that the microgrid could have a zero minimum
  ramp rate. Here, we discuss conditions under which the minimum ramp
  rate of the microgrid is non-zero. Let $ \E' = \{ e_j \in \E \;
  \vert \; \omega_j < \overline{\omega}_j \}$ be the set of all the
  lines which have not reached their flow limits when providing the
  maximum up regulation. Next, consider the graph $\G_m'=(\V,\E')$ and
  let $ \V'_g =\{ v_i \in \V_g \; \vert \; \exists \text{ a path from
  } i \text{ to } 1 \text{ in } \G_m \}$, i.e., the set of
  controllable nodes which are connected to the tie line. If $\V'_g
  \neq \phi$, then the minimum ramp rate is always non-zero. \hfill
  $\bullet$
\end{remark}

\subsection{Cost Function}\label{sec:cost}
Each aggregator needs to calculate the cost of providing a given
amount of regulation by capturing the effect of operating the
controllable nodes away from their baseline operating points.  
For an operating point~$g$, the total cost for the aggregator is given
by
\begin{equation}\label{eq:ag}
  h(g) = \sum_{p\in\V_g}h_{p}(g_{p}),
 \end{equation}
where $h_{p}:\real \to \real_{\ge 0}$ is the cost of operating node
$p$ away from its baseline level $g_{p}^0$. One representative example
of such a function is $h_p(g_p) = (g_p-g_p^0)^2$. The total regulation
that the aggregator provides is the combination of individual
regulations of controllable nodes. Therefore, for a specified
regulation level $x$, one would ideally choose the value of $g$ that
minimizes the total cost given by~\eqref{eq:ag} respecting the power
flow constraints in~\eqref{eq:pff} and the minimum and maximum
capacity constraints on each controllable node. Formally,
$f:[\overline{x},\underline{x}] : \real \to \real_{\ge 0}$, is given
by
\begin{equation}\label{eq:cost}
  f(x)= \begin{cases} \begin{aligned} & {\min_{g,\omega}}
      & &   h(g) \\
      & \text{s.t.}
      & & \underline{g} \leq g \leq \overline{g}\\
      &&& \begin{bmatrix}
        (P^0+ x)&
        \hspace{-1mm}g^{\top} &
        \hspace{-2mm}-l^{\top}
      \end{bmatrix}^{\top}\hspace{-2mm}=\M\omega\\
      &&&  |\omega| \leq \overline{\omega}.
    \end{aligned}
  \end{cases}
\end{equation}
However, a cost function defined like this does not take into account
the previous operating point of the microgrid and assumes that it can
transition between the optimal points corresponding to different
regulation powers arbitrarily fast. In practice, however, since the
regulation set points change every 2-4 seconds, ramp rates might limit
the change from optimal point at one time instant to the next. This
suggests that the cost of providing certain amount of regulation at
one instant also depends on the value of the regulation power at the
previous instant. Hence, we define the cost $\mathbf{f}:
[\overline{x},\underline{x}] \times [\overline{x},\underline{x}] \to
\real_{\ge 0}$, of providing regulation power $x$, if providing
regulation power $x^-$ at the previous instant, as
\begin{align}\label{eq:costnewno}
  & \min\limits_{g,\Delta g, \omega, \Delta \omega} & & h(g+ \Delta g)
  \notag
  \\
  & \quad \;\; \text{s.t.}  & & \underline{g} \leq g+\Delta g \leq
  \overline{g},\quad \Delta g \leq r
  \\
  &&&
  \begin{bmatrix}
    (P^0+ x)& \hspace{-1mm}(g+ \Delta g)^{\top} &
    \hspace{-2mm}-l^{\top}
  \end{bmatrix}^{\top}\hspace{-2mm}=\M(\omega+\Delta \omega)\nonumber
  \\
  &&&  |\omega+ \Delta \omega| \leq \overline{\omega} \nonumber
  \displaybreak[0]  
  \\
  &&& \underline{g} \leq g \leq \overline{g} \quad |\omega| \leq
  \overline{\omega} \nonumber
  \\
  &&& \begin{bmatrix} (P^0 + x^-) & g^{\top} & \hspace{-2mm}-l^{\top}
  \end{bmatrix}^{\top}=\M\omega. \nonumber 
\end{align}
Here, $(g+\Delta g,\omega+\Delta \omega)$ and $(g,\omega)$ are the
vectors of the power levels of controllable nodes and line flows when
the microgrid provides regulation power $x$ and $x^-$,
respectively. The constraints also enforce the capacity limits for the
individual controllable nodes and the flow limit constraints for both
values of regulation power, and the ramp constraints in transitioning
from $x^-$ to~$x$.  The reason to include the power flow constraints
at $x^-$ in~\eqref{eq:costnewno} is to enable the aggregator to
pre-compute the cost function independently of the regulation power it
might be asked to provide.  Otherwise, if the cost is computed at
every regulation instant, $g$ and $\omega$ providing $x^-$ would be
known, and the optimization variables would only be $\Delta g$, and
$\Delta \omega$. As such, $\mathbf{f}(x,x^-)$ is a lower bound on the
actual cost since $(g,\omega)$ are also decision variables and are
selected optimally to move to the next operating point.

The following result, whose proof is given in the Appendix, identifies
a condition that simplifies the computation of the cost function
$\mathbf{f}(x,x^-)$ defined in~\eqref{eq:costnewno}.

\begin{lemma}\longthmtitle{Simplified formulation and convexity of 
    cost function}\label{lemma:cvx}
  Given regulation powers $x^-$ and $x$, if $|x-x^-| \leq R(x^-)$,
  then $\mathbf{f}(x,x^-)= f(x)$.   
If $h$ is (strictly) convex, then $f$ is (strictly) convex.
\end{lemma}

Figure~\ref{fig:cost} provides the cost functions~\eqref{eq:cost} of
the four groups of microgrids displayed in Figure~\ref{fig:capacity}
in the constant load case.

\begin{figure}[htb]
  \centering
  \includegraphics[scale=0.45]{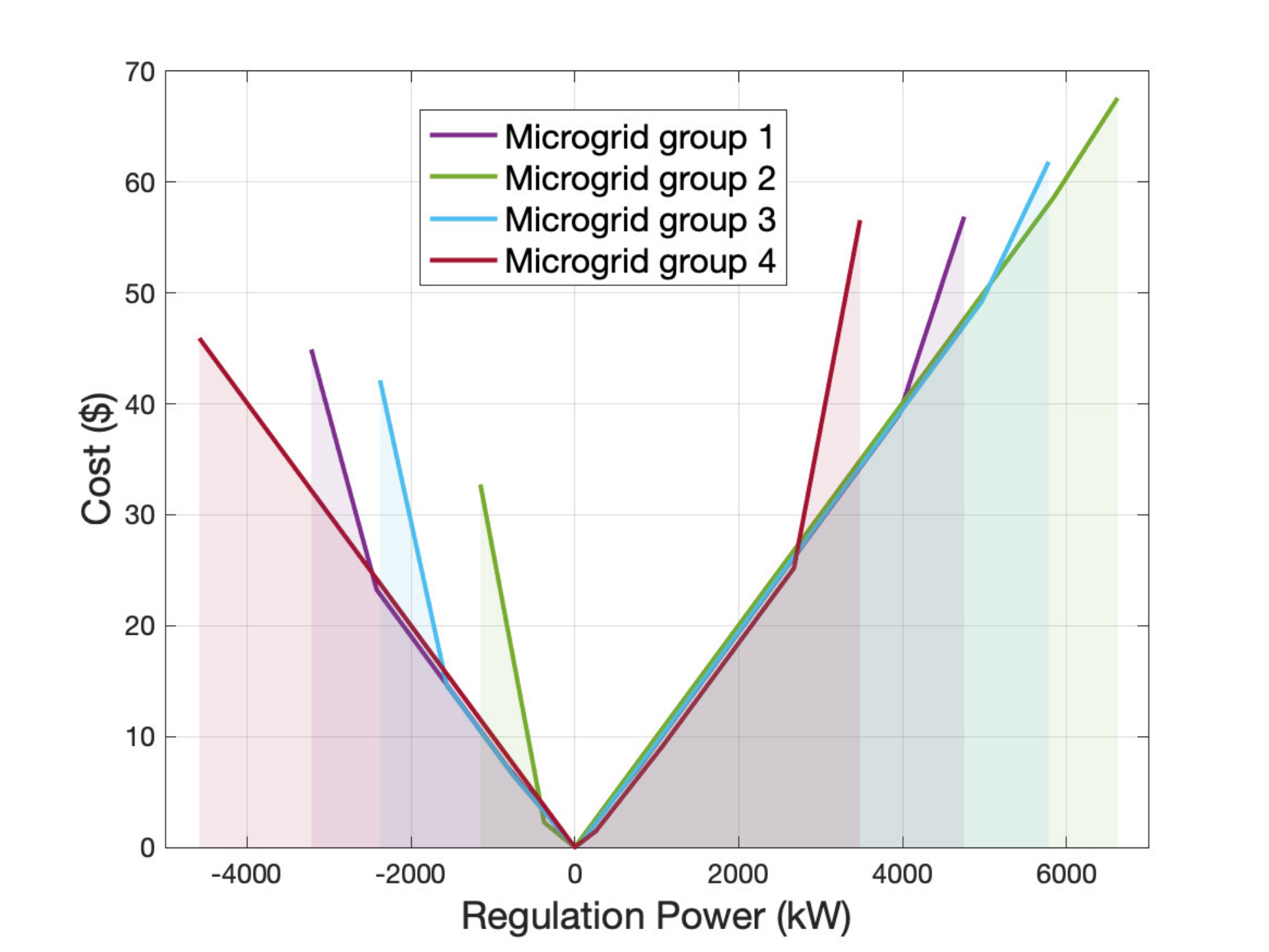}
  \caption{Abstracted cost functions for different instantiations of
    the reduced-order UCSD microgrid with constant loads. The shaded
    regions represent the range of regulation power that the
    corresponding microgrid can provide.}
 \label{fig:cost}
\end{figure}

Note that the cost function~\eqref{eq:cost} assumes the load to be
constant, but since the aggregator is not required to submit its cost
functions in [CP1], there is no need to pre-compute this using
probabilistic techniques.  Instead, the cost function at a given
regulation instant could be computed online using the load at that
instant.  The time taken to compute the cost function at a given
instant would depend upon the type of solver used, but is usually
small (e.g., less than a second with built in MATLAB solver
\verb|fmincon|). In addition,
since the regulation period lasts for 10-15 minutes, the variation in
load would be limited, thereby requiring the recomputation of the cost
function sparingly.

\subsection{Bids for Participation in Market Clearance}
Based on the abstractions in Sections~\ref{sec:bounds}-\ref{sec:cost},
here we specify the bid information used by each aggregator to
participate in [CP1].  Without loss of generality, we specify the bid
quantities for up regulation market.  Let $g^{\text{up}} \in
\real^{n_g}$ denote the component in $g$ of the solution
of~\eqref{eq:upreg}.  

\begin{table}[h]
  \centering
  \caption{Bidding quantities for up regulation market}\label{tab}
  \begin{tabular}{|l|l|}
    \hline
    Bid Quantity   & Value
    \\ \hline \hline 
    Capacity       & $|\overline{x}|$ 
    \\ \hline
    Mileage        & $k  \R(g^{\text{up}})$
    \\ \hline 
    Capacity price & $h(g^{\text{up}})/|\overline{x}|$
    \\ \hline 
  \end{tabular}
\end{table}

Table~\ref{tab} specifies the proposed values for the bidding
quantities.  Here $k > 0$ is a constant depending on the duration of
the regulation period and update frequency of the AGC setpoints.
The suggested bids are conservative, meaning that the aggregator would
be able to provide whatever it promises, and there is no strategy to
maximize profit.  It might seem from Table~\ref{tab} that there is no
need to compute beforehand the whole ramp rate function $\R$ in
Section~\ref{sec:ramp}. However, a risk taking aggregator might use a
higher value of mileage bid based on the shape of~$\R$.  It is also
interesting to note that, from the convexity of cost function in
Lemma~\ref{lemma:cvx} and the capacity price bid in Table~\ref{tab},
the aggregator would never be at loss regardless of the regulation
power being provided.

\section{RTO-DERP Coordination Problem}\label{sec:rto-derp}
Here we describe our algorithmic solution for the RTO-DERP
coordination problem [P2] to disaggregate the regulation
signal. Equipped with the microgrids' capacities and cost and ramp
rate functions identified in Section~\ref{sec:abstractions}, the
aggregators, communicating over a graph $\G$, seek to solve, at each
instant of the regulation period, the optimization
problem~\eqref{eq:iso}. However, as we have noted before, this problem
might not always be feasible due to the presence of ramp constraints.
This means that in principle, at each regulation instant, one would
need to solve~\eqref{eq:iso} if it is feasible or minimize the
difference between the required regulation and the procured regulation
if it is infeasible. Such dichotomy also raises the issue of the
necessary information available to the aggregators to determine which
one of the two cases to address at each regulation instant.

Instead, we propose to reformulate the optimization problem in a way
that lends itself to the identification of solutions that minimize the
error between the procured regulation and the required regulation
whenever \eqref{eq:iso} is not feasible.
Without loss of generality, throughout this section we assume the
required regulation power to be positive. 
We start by defining the problem
\begin{equation}\label{eq:isonew}
  \begin{aligned}
    & {\min_{x}}
    & &  f^\mu(x)=f(x)+\mu [\Delta x ]^+ \\
    & \text{s.t.} & & \underline{x_{\im}} \leq x_{\im} \leq
    \overline{x_{\im}} \quad \forall \im
    \\
    &&& |x_{\im} - x_{\im}^-| \leq R_{\im}(x_{\im}^-) \quad \forall \im,
  \end{aligned}
\end{equation}
where $\mu > 0$ is a penalty parameter and $\Delta x=x_r - \one^{\top}
x$.  The following result, whose proof is given in the Appendix,
characterizes the equivalence between problems~\eqref{eq:isonew}
and~\eqref{eq:iso}.

\begin{lemma}\longthmtitle{Equivalence between~\eqref{eq:iso}
    and~\eqref{eq:isonew}}\label{lemma:equivalence}
  Optimization~\eqref{eq:isonew} is always feasible and there exists
  $\hat{\mu} < \infty$ such that for all $\mu \in [\hat{\mu},
  \infty)$, ~\eqref{eq:iso} and~\eqref{eq:isonew} have the same
  solution set if~\eqref{eq:iso} is feasible.
\end{lemma}

\begin{remark}\longthmtitle{Establishing the threshold value
    $\hat{\mu}$ without the knowledge of dual optimizers}
  The threshold value $\hat{\mu}$ in Lemma~\ref{lemma:equivalence}
  depends on the optimal values of the dual variables, which is not
  known beforehand.  Interestingly, the explicit knowledge of the
  Lagrange multipliers to obtain a lower bound on the value of $\mu$
  can be avoided. In fact, according to~\cite[Proposition
  5.2]{AC-JC:15-tcns}, we have
  \begin{align} 
    \hat{\mu} \ge 2 \max\limits_{x \in \mathcal{F}}\|\nabla f
    (x)\|_\infty . \tag*{$\bullet$}
  \end{align}
 \end{remark}

Given Lemma~\ref{lemma:equivalence}, we focus on solving
problem~\eqref{eq:isonew} in a distributed way.  To handle the local
constraints, we again reformulate~\eqref{eq:isonew} using exact
penalty function as
\begin{align}\label{eq:isonewlocal}
  \min\limits_{x} \;f^p(x)=& \underbrace{ f(x) + \mu_2
    \sum\limits_{\im=1}^N \big( [\overline{b_{\im}}]^+ +
    [\underline{b_{\im}}]^+ \big)}_{f^{\mu_2}(x)}+ \mu [\Delta x]^+,
  \\
  \text{where } \overline{b_{\im}} =& x_{\im} -
  \min\{\overline{x_{\im}},x_{\im}^- + R_{\im} (x_{\im}^-) \}, \notag
  \\
  \text{and }\underline{b_{\im}} =&
  \max\{\underline{x_{\im}},x_{\im}^- - R_{\im} (x_{\im}^-)\} -
  x_{\im}, \notag
\end{align}
are the box constraints taking care of the capacity and ramp rate for
aggregator $\im \in \{1,\ldots,N\}$ and $\mu_2 > 0$ is again a penalty
parameter.  Once again, similar to Lemma~\ref{lemma:equivalence},
there exist finite values of $\mu_2$ for which the
reformulation~\eqref{eq:isonewlocal} is exact.

Since problem~\eqref{eq:isonewlocal} is unconstrained, consider the
dynamics
\begin{equation}\label{eq:gradientdescent}
  \dot{x} \in -\partial f^p(x),
\end{equation}
where $\partial f^p : \real^N \rightrightarrows \real^N$ denotes the
generalized gradient of $f^p$.  For each agent $\im \in
\{1,\ldots,N\}$, $[\partial f^p(x)]_{\im}$ is given by
\begin{align*}
  \begin{cases}
    \nabla f_{\im}(x_{\im}) - [\mu]^+_{\Delta x} -
    [\mu_2]^+_{\overline{b_i}} + [\mu_2]^+_{\underline{b_i}}, &\Delta
    x, \overline{b_{\im}},\underline{b_{\im}} \neq 0,
    \\
    \nabla f_{\im}(x_{\im}) - [0, \mu]- [\mu_2]^+_{\overline{b_i}} +
    [\mu_2]^+_{\underline{b_i}} , &\Delta x = 0,
    \overline{b_{\im}},\underline{b_{\im}} \neq 0,
    \\
    \nabla f_{\im}(x_{\im}) - [\mu]^+_{\Delta x} - [0,\mu_2] +
    [\mu_2]^+_{\underline{b_i}}, &\Delta x, \underline{b_{\im}} \neq
    0,\overline{b_{\im}} = 0,
    \\
    \nabla f_{\im}(x_{\im}) - [\mu]^+_{\Delta x} -
    [\mu_2]^+_{\overline{b_{\im}}} + [0,\mu_2], &\Delta x,
    \overline{b_{\im}} \neq 0,\underline{b_{\im}} = 0,
    \\
    \nabla f_{\im}(x_{\im}) - [0, \mu]- [0,\mu_2] +
    [\mu_2]^+_{\underline{b_i}} , &\Delta x, \overline{b_{\im}}
    =0,\underline{b_{\im}} \neq 0,
    \\
    \nabla f_{\im}(x_{\im}) - [0, \mu]- [\mu_2]^+_{\overline{b_i}} +
    [0,\mu_2] , &\Delta x , \underline{b_{\im}}=0,\overline{b_{\im}}
    \neq 0.
  \end{cases}
\end{align*} 
The equilibria of the dynamics~\eqref{eq:gradientdescent} satisfy
$\zero \in \partial f^p(x)$.  Asymptotic convergence
of~\eqref{eq:gradientdescent} to the optimizers
of~\eqref{eq:isonewlocal} could be easily established using tools from
non-smooth analysis, cf.~\cite[Proposition 14]{JC:08-csm}.
However, the implementation of~\eqref{eq:gradientdescent} requires
every aggregator to have the knowledge of the total regulation at all
times. To handle this, we use dynamic average consensus,
cf. Section~\ref{sec:prelims}, to estimate the average of the
difference between the required regulation and procured regulation
from all the microgrids.  Since $\frac{1}{N}\Delta x$ and $\Delta x$
have the same signs, we modify algorithm~\eqref{eq:gradientdescent}
using dynamic average consensus as follows
\begin{subequations}\label{eq:gd+dac}
  \begin{align}
    \!\! \dot{x} & \in -\partial f^{\mu_2} (x) +
    [\mu]_{z}^+ \label{eq:x},
    \\
    \!\! \dot{z}& \in -\nu z -\beta \Lap z - v + \nu (x_r e - x)
    + \partial f^{\mu_2}(x) - [\mu]^+_z,\label{eq:z}
    \\
    \!\!\dot{v}&=\nu \beta \Lap z, \label{eq:v}
  \end{align}
\end{subequations}
where $\Lap \in \real^{N \times N}$ is the Laplacian matrix of $\G$,
$z, v \in \real^N$, $z_{\im}$ is the $\im$th aggregator's estimate of
$\frac{1}{N}\Delta x$, $[\mu]^+_z \in \real^N$ with its $\im$th
element as $[\mu]^+_{z_{\im}}$ and $e$ is the unit vector with only
one entry as one and all others as zero. Note immediately that the
algorithm~\eqref{eq:gd+dac} is distributed over the communication
graph and only one aggregator needs to know the required regulation.
We refer to~\eqref{eq:gd+dac} as ``gradient descent + dynamic average
consensus" algorithm, abbreviated as $\gdac$.
The equilibria for $x$ are the points satisfying $\zero \in -\partial
f^{\mu_2}(x)+ [ \zero, \mu \one]$. The next result, whose proof is
given in the Appendix, characterizes the convergence properties of the
$\gdac$ algorithm.

\begin{theorem}\longthmtitle{Asymptotic convergence of the distributed
    dynamics to the optimizers} \label{thm:gdac}
Let $\G$ be strongly connected and weight-balanced, and the initial conditions satisfy $\one^{\top} v(0)=0$ and
  $\one^{\top} z(0) - \Delta x(0)=0$, then there exists $\bar{\mu} <
  \infty$ such that the dynamics~$\gdac$ find the optimizers of~\eqref{eq:isonewlocal} for all $\mu \in [\bar{\mu},
  \infty)$.
\end{theorem}

\begin{remark}\longthmtitle{Initialization of the distributed
    algorithm}
  For the dynamics~$\gdac$ to converge to the optimizers,
  Theorem~\ref{thm:gdac} specifies requirements on the initial
  conditions. The requirement $\one^{\top} v(0)=0$ could be
  implemented trivially by selecting $v(0)=\mathbf{0}$. For the
  implementation of $\one^{\top} z(0) - \Delta x(0) = 0$, the
  aggregators can simply choose $z(0)=\zero$ and $x_{\im}(0)=0$ for
  all $\im$, except for the aggregator who knows the required
  regulation for which $x_{\im}(0)=x_r$. \hfill $\bullet$
\end{remark}

\section{Simulations}\label{sec:sims}

\begin{figure*}[ht]
  \centering
  \includegraphics[width=.95\textwidth]{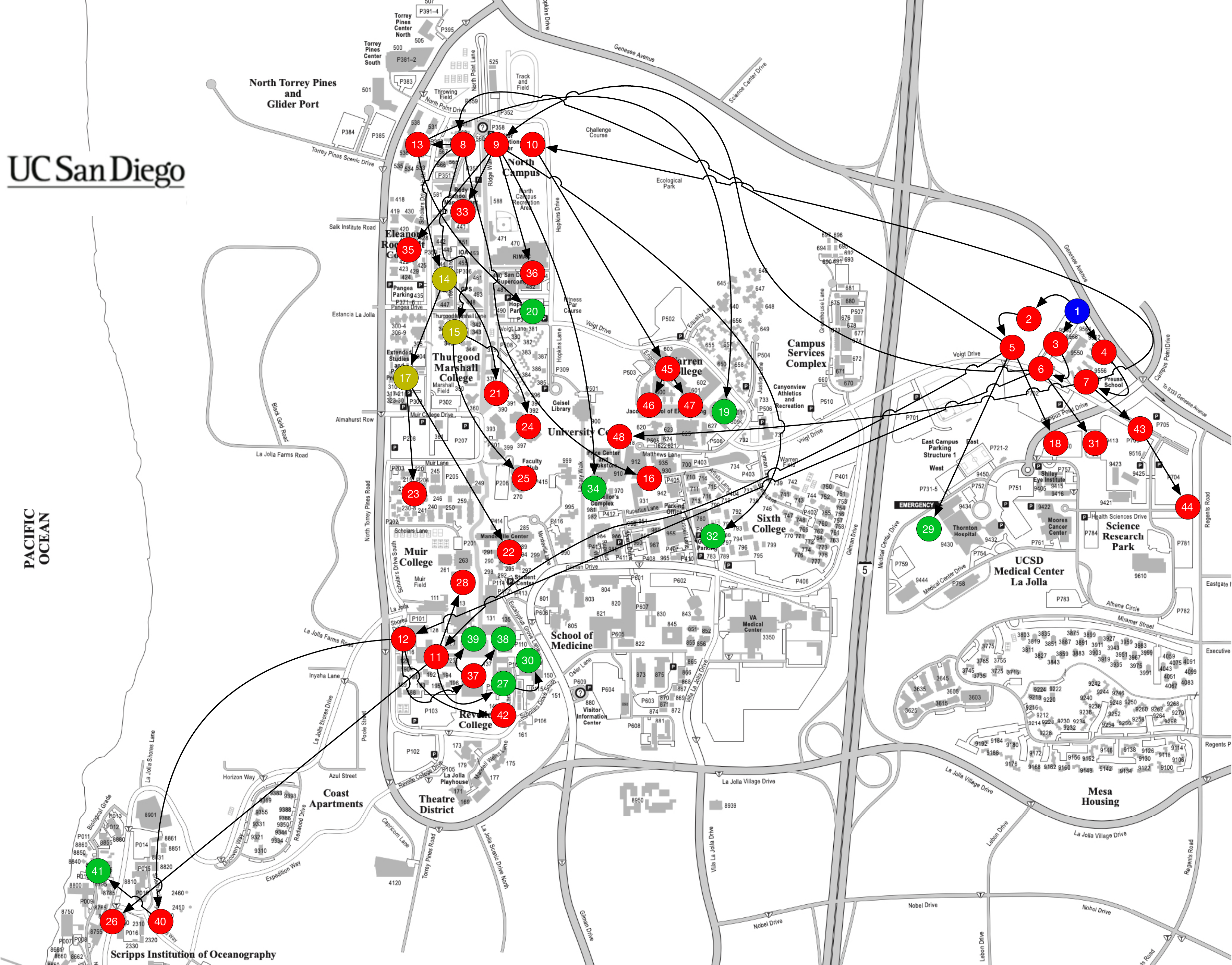}
  \caption{Reduced-order model of the UCSD microgrid.  Arrowheads
    represent the direction of positive flows.  Blue node (1) is connected to
    the tie line. Green nodes (19, 20, 27, 29, 30, 32, 34 38, 39 and
    41) represent the generators, dark yellow (14, 15 and 17) the electric
    vehicle stations and red (remaining) the building
    loads.}\label{fig:map}
  \vspace*{-1ex}
\end{figure*}

We provide here our simulation results based on the abstractions of
capacities, cost, and ramp rate developed in
Section~\ref{sec:abstractions} and the RTO-DERP coordination
algorithm~\eqref{eq:gd+dac} in Section~\ref{sec:rto-derp}. For the
purpose of simulations, we consider a reduced-order model of the
University of California, San Diego (UCSD) microgrid developed using
the distributor feeder reduction algorithm in~\cite{ZKP-VRD-MJR-JK:18}
and provided by the research group of Prof. Jan Kleissel.  Compared to
the full-order model of the UCSD
microgrid~\cite{BW-JD-DW-JK-NB-WT-CR:13} which is a radial, balanced
network with 1289 buses (3869 nodes), the reduced-order model is a
balanced tree network with 48 buses. The buses in the reduced-order
model are obtained by retaining the key buses in the full order model
which are the buses where the building loads aggregate or which have
generators.  Since the UCSD reduced-order model is balanced, we
consider only one phase in our simulations. The model consists of 10
generators (2 gas turbines, 1 steam turbine, and 7 solar PV systems)
and 37 loads (34 building loads and 3 electric vehicle stations). We
show the location of the buses on the geographical map of the campus
in Figure~\ref{fig:map}.  For our simulation, we take the UCSD
microgrid as a template, and we instantiate it using different
baseline scenarios to construct 12 different microgrids, divided into
4 groups. The baseline values of generation and mean load is constant
within a group. The 3 different scenarios within a group consist of
(a) constant load, (b) variable load with failure probabilities
$\epsilon'=10^{-1}$, $\epsilon=4.2 \times 10^{-5}$ and (c) variable
load with failure probabilities~$\epsilon'=2 \times 10^{-1}$,
$\epsilon=8.4 \times 10^{-5}$. The abstracted regulation capacities
and ramp rate functions of different microgrid groups are shown in
Figures~\ref{fig:capacity} and~\ref{fig:ramp}, resp.  For cost
functions, we consider quadratics for all the resources.  The
abstracted cost functions for different groups are shown in
Figure~\ref{fig:cost}.

We demonstrate the performance of the distributed
algorithm~\eqref{eq:gd+dac} 
in two sets of simulations.  To implement the continuous-time
algorithm, we use a first-order Euler discretization with step size of
0.001 to show its practical feasibility.  The values of $\mu$,
$\mu_2$, $\beta$ and $\nu$ are taken to be 1000, 1100, 400 and 400,
respectively.  In the first simulation,
cf. Figure~\ref{fig:one_instant}, we consider one regulation instant
and first show the evolution of the proposed algorithm for required
down regulation of 50000~kW, and compare it, for the same
communication topology (undirected ring with few additional edges),
against the (2-hop distributed) saddle-point
dynamics~\cite{AC-BG-JC:17-sicon} of the augmented Lagrangian for the
equivalent reformulated problem as per~\cite{AC-JC:16-allerton}. As
can be seen from the plots, the algorithm time required by the
proposed algorithm to reach~1\% band of the required regulation power
in much less compared to the saddle-point dynamics.
The time required does increase when the communication topology is
changed to a directed ring - which is the worst possible topology for
strongly connected graphs, but still remains less than a second,
implying that the number of iterations is less than 1000.

\begin{figure}[htb]
  \centering
  \includegraphics[scale=0.45]{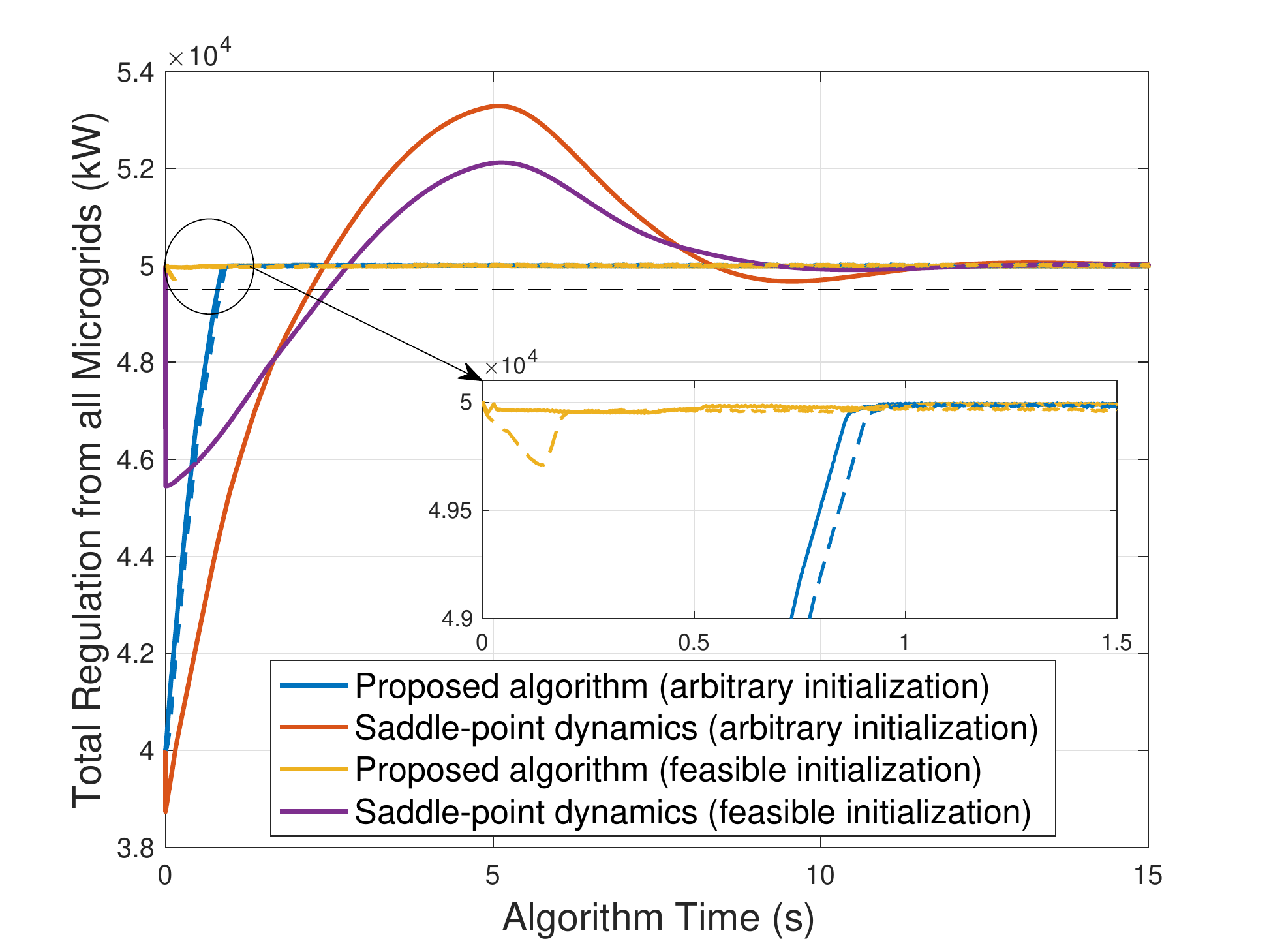}
  \caption{Performance of the proposed distributed algorithm compared
    against the saddle-point dynamics for 50000~kW down regulation
    from 12 aggregators.  The dashed lines for the proposed algorithm
    represent the algorithm evolution over a directed ring, and the
    black dashed lines represent 1\% band of the required regulation
    power. Dynamics were implemented in discrete time with a step size
    of 0.001 and the values of $\mu$, $\mu_2$, $\beta$ and $\nu$ as
    1000, 1100, 400 and 400, respectively.}\label{fig:one_instant}
  \vspace*{-1ex}
\end{figure}

For the second simulation, we consider the dynamic regulation test
signal (RegD), available on the Pennsylvania-New Jersey-Maryland
Interconnection (PJM) website~\cite{PJM}. Since the RegD signal on the
PJM website is normalized and could be scaled as long as the problem
remains feasible, we scale it by a factor of 50000 and then use our
abstractions and clear the market according to [CP1]. Once the market
is cleared, we use our algorithm to track the scaled RegD signal and
compare it using the current algorithm of disaggregating the
regulation signal described in [CP2]. For the sake of clarity, we show
only the first 100 instants of the regulation period, and instead of
contributions from each of the 12 microgrids, show the total
contributions from the 4 groups. As we can see from
Figure~\ref{fig:compare}, when it is not possible to provide the
required amount due to limits on ramp rates, both the proposed
algorithm and the current algorithm try to provide as much regulation
power as possible, and the tracking performance for both the
algorithms is similar.  But, if we compare the cost, the proposed
algorithm with a cost of \$9104 outperforms the current algorithm with
a cost of \$9638. This difference in cost comes from very different
power contributions from the microgrids for the two algorithms. The
proposed algorithm allocates the regulation signal to the microgrids
based on their abstracted cost functions (cf.~Figure~\ref{fig:cost}),
whereas current practice does not take them into account.  It can be
noticed in Figure~\ref{fig:compare} that, under current practice, if
not capped by the cleared capacities, the power allocations for
different microgrid groups have the same ratios for every regulation
instant.  For example, the shape of the regulation power curves for
microgrid groups 1, 3 and 4 are similar and only differ in terms of
scaling (by factors depending on the ratio of their procured
mileages).

\begin{figure}[htb]
  \centering
  \includegraphics[scale=0.45]{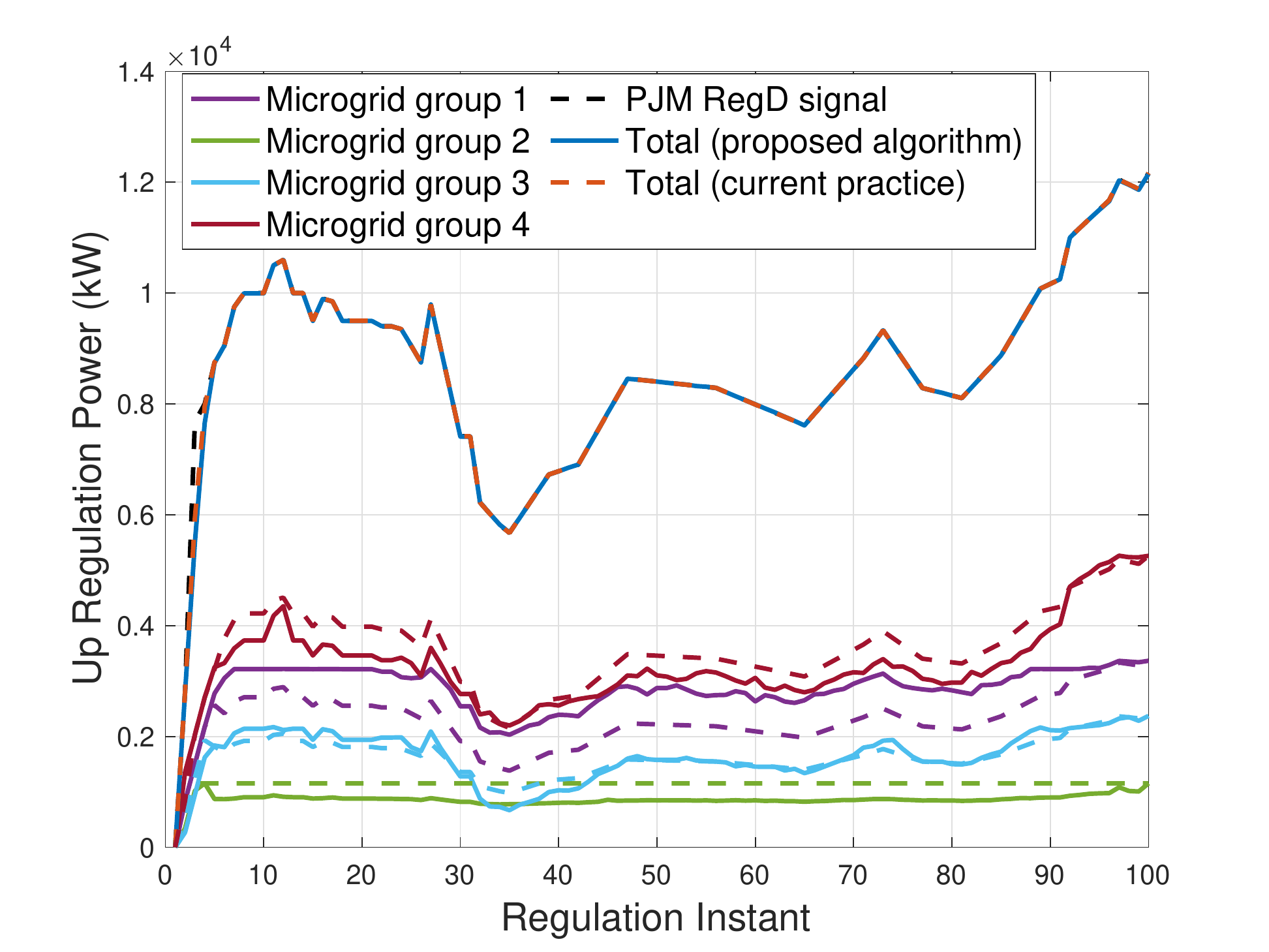}
  \caption{Performance of the proposed RTO-DERP coordination algorithm
    and the algorithm followed currently tested against first 100
    updates of the PJM RegD signal.  For microgrid groups, the solid
    lines represent the regulation power using the proposed algorithm
    and the dashed lines represent the regulation power using current
    practice.  Although the tracking performance for both the
    algorithms is similar, contributions from individual microgrids
    differ substantially resulting into different
    costs.}\label{fig:compare}
  \vspace*{-1ex}
\end{figure}

\section{Conclusions}
We have considered the problem of providing frequency regulation
services by aggregations of DERs. We have described the limitations of
current practice and identified the challenges to overcome them with
DER aggregators modeled as microgrids.  We have developed meaningful
abstractions for the capacity, cost of generation, and ramp rates by
taking into account the power flow equations inside the
microgrid. This provides enough information for the microgrids to
participate in the market clearance stage.  We have employed these
abstractions to design a provably correct distributed algorithm that
solves the RTO-DERP coordination problem to optimally disaggregate the
regulation signal when the problem is feasible and minimize the
difference between the required regulation and procured regulation
when it is infeasible. Future work will extend our work to microgrids
with more general topologies, incorporate AC power flow equations, and
investigate smooth distributed algorithms to remove any chattering due
to non-smooth dynamics.

\appendix
Here we provide proofs of all the results stated in the paper.

\subsection*{Proof of Lemma \ref{lemma:chance to linear}}
With the notation of the statement,~\eqref{eq:pseudoinv} can be
written as
\begin{align*}
  \begin{vmatrix}
    \M_1 (\mathbf{1}^\top l - \mathbf{1}^\top g) + \M_2 g - \M_3 l + \N
    \gamma
  \end{vmatrix} \leq \overline{\omega}.
\end{align*}
Without loss of generality, let us for now consider only the following
constraint in~\eqref{eq:upregchance}
\begin{align}\label{eq:positiveconstraints}
  \Pr (|\zeta_j|-\overline{\omega}_j \leq 0) \geq 1- \epsilon.
\end{align}
where $\zeta_j=(\M_{1j} \mathbf{1}^\top - \M_{3j}) l + (\M_{2j} - \M_{1j}
\mathbf{1}^\top) g + \N_j \gamma $.
Let $\xi_j^+=\setdef{\zeta_j \in \real}{\zeta_j-\overline{\omega}_j \leq 0}$ and
$\xi_j^-=\setdef{\zeta_j \in \real}{-\zeta_j-\overline{\omega}_j \leq 0}$.
Then~\eqref{eq:positiveconstraints} is equivalent to
\begin{align}\label{eq:pr_int}
  \Pr ( \xi_j^+ \cap \xi_j^- ) \geq 1- \epsilon,
\end{align}
We can further rewrite~\eqref{eq:pr_int} as
\begin{align}
\Pr (  \xi_j^+ \cap \xi_j^-)^c \leq \epsilon
\Rightarrow \Pr(\xi_j^{+c} \cup \xi_j^{-c}) \le \epsilon
. \label{eq:pr_complement}
\end{align}
We next break~\eqref{eq:pr_complement} 
down into single chance constraints. 
Using the fact that $\xi_j^{+c}$ and $\xi_j^{-c}$ are mutually exclusive, 
$\Pr(\xi_j^{+c} \cup \xi_j^{-c})=\Pr(\xi_j^{+c}) +\Pr( \xi_j^{-c})$. Therefore,~\eqref{eq:pr_complement}
is equivalent to
\begin{align}\label{eq:pr_alllines}
   \Pr (\xi_j^{+c}) \leq \epsilon/2, \text{ and } \Pr(\xi_j^{-c}) \le& \epsilon/2.
 \end{align}
If $l \sim \mathcal{N}( \hat{l}, V_l)$, then $\zeta_j \sim
\mathcal{N}(\hat{\zeta_j}, \sigma_j^2)$ where
\begin{align*}
  \hat{\zeta_j} &=(\M_{1j} \mathbf{1}^\top - \M_{3j}) \hat{l} + (\M_{2j}
  - \M_{1j} \mathbf{1}^\top) g + \N_j \gamma,
  \\
  \sigma_j^2 &= (\M_{1j} \mathbf{1}^\top - \M_{3j}) V_l (\M_{1j}
  \mathbf{1}^\top - \M_{3j})^\to.
\end{align*}
Defining $w = {(\zeta_j-\hat{\zeta_j})}/{\sigma_j}$, we have $w \sim
\mathcal{N}(0,1)$.
\begin{align}\label{eq:phi}
  \Pr (\xi_j^{+}) = \Pr \Big(w \leq \frac{\overline{\omega}_j-\hat{\zeta_j}}{\sigma_j}
  \Big) = \Phi \Big(\overline{\omega}_j-\frac{\hat{\zeta_j}}{\sigma_j} \Big).
\end{align}
Using equations~\eqref{eq:phi} and~\eqref{eq:erf}, we have
from~\eqref{eq:pr_alllines} for $\Pr(\xi_j^{+})$
\begin{align*}
  &\frac{1}{2}+\frac{1}{2} \erf \Big( \frac{\overline{\omega}_j -
    \hat{\zeta_j}}{\sqrt{2} \sigma_j} \Big) \geq 1-\epsilon/2 ,
  \\
  & \Rightarrow \erf \Big(\frac{\overline{\omega}_j -
    \hat{\zeta_j}}{\sqrt{2} \sigma_j} \Big) \geq 1-\epsilon ,
  \\
  & \Rightarrow \hat{\zeta_j} \leq \sqrt{2} \sigma_j \erf^{-1}
  (\epsilon -1) + \overline{\omega}_j.
\end{align*}
A similar inequality could be obtained from~\eqref{eq:pr_alllines} for
$\Pr(\xi_j^-)$.  As a result,~\eqref{eq:pr_alllines} could be
rewritten as
\begin{align*}
  |\hat{\zeta_j}| \leq \sqrt{2} \sigma_j \erf^{-1} (\epsilon -1) +
  \overline{\omega}_j.
\end{align*}
The righthand side of the above constraint is a constant dependant on
$\epsilon$ and the left hand side depends on the decision variables
$g$ and $\gamma$.

The same technique could be applied to the remaining set of
constraints. If we apply this to all the chance constraints
in~\eqref{eq:upregchance}, then problem~\eqref{eq:upregchance} could
be solved by solving the deterministic linear
program~\eqref{eq:chance-det}. \hfill $\blacksquare$

\begin{lemma}\longthmtitle{Simplified power flow constraints for tree
    network}\label{lemma:simplify}
  Let $\G_m$ be a tree and $\PP_{\refs} \in \real^{(n-1) \times
    (n-1)}$ denote its path matrix with first vertex as reference
  $\refs$.  Then the constraints
  \begin{subequations}\label{eq:constraints_ramp}
    \begin{align} 
      \begin{bmatrix}\label{eq:constraints_ramp:a}
        (P-\mathbf{1}^\top \Delta g) & (g+\Delta g)^\top & -l^\top
      \end{bmatrix}^\top =&\M(\omega+ \Delta \omega), 
      \\
      |\omega+\Delta \omega| \leq  & \;
      \overline{\omega},
    \end{align}
  \end{subequations}
  in~\eqref{eq:rc} could be equivalently written as
  \begin{align}\label{eq:cd}
    \PP_1^\top\Delta g \leq \overline{\omega} + \PP_2^\top l- \PP_1^\top g,
  \end{align}
  where $[\PP_1^\top \quad \PP_2^\top] = |\PP^\top_{\refs}|$,
  with $\PP_1 \in \real^{ n_g \times (n-1)}$ and $\PP_2 \in \real^{n_l
    \times (n-1) }$, and $|\PP^\top_{\refs}|$ denotes the non-negative
  matrix whose elements are given by the absolute values of the
  corresponding elements of $\PP^\top_{\refs}$.
\end{lemma}
\begin{proof}
  Let $\M_{\refs} \in \real^{(n-1) \times (n-1)}$ denote the matrix
  obtained after removing the row corresponding to vertex $\refs$ from
  $\M$.  According to~\cite{JR:63}, we have
  \begin{equation}\label{eq:pm}
    \M_{\refs}^{-1}=\PP_{\refs}^\top.
  \end{equation}
  With first vertex as $\refs$, equation~\eqref{eq:pffa} could be
  rewritten as
  \begin{equation}\label{eq:pfr}
    \begin{bmatrix}
      g+ \Delta g
      \\
      -l
    \end{bmatrix}
    = \M_{\refs} (\omega+ \Delta \omega),
  \end{equation}
  where we have used the fact that $\rank(\M)=\rank(\M_{\refs}) =n-1$,
  cf.~\cite[Corollary 4-4]{SS-MBR:61}.  Using~\eqref{eq:pfr}
  and~\eqref{eq:pm}, constraint~\eqref{eq:constraints_ramp} is
  equivalent to
  \begin{align*}    
-\overline{\omega} \leq \PP_{\refs}^\top \begin{bmatrix}g+\Delta
      g \\ -l\end{bmatrix} \leq \overline{\omega}.
  \end{align*}
  Due to the structure of $\PP_{\refs}$,
  cf. Section~\ref{sec:prelims}, all the non-zero entries for any row
  of $\PP^\top_{\refs}$ are either 1 or -1.  Since we are
  characterizing the ramp up rate and are only concerned with what
  happens to the feasible region with the increase in some
  component(s) of $g$, the active constraint for the lines for which
  the non-zero entries are 1 would be
  \begin{subequations}\label{eq:cdpm}
    \begin{align}
      \PP_{\refs}^\top \begin{bmatrix} g+ \Delta g \\ -l\end{bmatrix}
      \leq \overline{\omega},
    \end{align}
    and for the lines for which the non-zero entries are -1 would be
    \begin{align}
      -\PP_{\refs}^\top \begin{bmatrix} g+ \Delta g \\ -l\end{bmatrix}
      \leq \overline{\omega}.
    \end{align}
  \end{subequations}
  \eqref{eq:cdpm} is equivalent to~\eqref{eq:cd}, completing the proof.
\end{proof}

\subsection*{Proof of Proposition \ref{prop1}}
Let us start by denoting the region where 
\begin{align*}
\PP_1^\top r \leq \overline{\omega} +
\PP_2^\top l -\PP_1^\top g,
\end{align*}
by $V_1$.
Boundaries of $V_1$ are $(n-1)$ hyperplanes given by
\begin{align*}
 \PP_1^\top r = \overline{\omega} +
\PP_2^\top l-\PP_1^\top g.
\end{align*}
Some of these hyperplanes could
even be outside $H$. But in general, all these $(n-1)$
hyperplanes could be the faces of $V_1$. It is clear that in $V_1$,
none of the flow constraints is active and $\R(g)=\mathbf{1}^\top r$.

Outside $V_1$, we have 
\begin{align}\label{eq:vnot}
\overline{\omega}_j + \PP_{2j}^\top l-\PP_{1j}^\top g < \PP_{1j}^\top r
\end{align}
for at least one $j \in \{1, \ldots, n-1\}$. First we consider the
region where~\eqref{eq:vnot} holds for only one such $j$, denoted as
$j'$. Then either
\begin{align*}
 \overline{\omega}_{j'} + \PP_{2j'}^\top l-\PP_{1j'}^\top g > 0,
\text{ or } \overline{\omega}_j' + \PP_{2j'}^\top l-\PP_{1j'}^\top g =0.
\end{align*}
 In the
former case, we are in the polyhedron whose two faces are given by
\begin{align*}
\overline{\omega}_{j'} + \PP_{2j'}^\top l-\PP_{1j'}^\top g = \PP_{1j'}^\top r,  \text{ and }
 \overline{\omega}_{j'} + \PP_{2j'}^\top l-\PP_{1j'}^\top g =0. 
\end{align*}
Let us call one of these polyhedron $V_2$. In $V_2$,
$\R(g)=\mathbf{1}^\top \Delta g$, where $\Delta g$ satisfies 
\begin{align*}
\overline{\omega}_{j'} + \PP_{2j'}^\top l-\PP_{1j'}^\top g = \PP_{1j'}^\top \Delta g.
\end{align*}
For $\mathbf{1}^\top \Delta g$ to be maximum, the controllable nodes for
which the corresponding entries are zero in $\PP_{1j'}$, we will have
$\Delta g_p=r_p$.
As some component(s) of $g$ for which the corresponding entry in $\PP_{1j'}=1$ increases,
some components of $\Delta g$ with corresponding entry 1, decrease to balance it. Hence,
$\R(g)=\mathbf{1}^{\top}r-\PP_{1j'}^{\top} g$. 
Now considering the latter
case when 
\begin{align*}
\overline{\omega}_j' + \PP_{2j'}^{\top} l- \PP_{1j'}^{\top} g =0.
\end{align*}
On this hyperplane, $\R(g)$ becomes constant again as the controllable
nodes for which the corresponding entries are zero in $\PP_{1j'}$ have
$\Delta g_p=r_p$ and other entries of $\Delta g$ have to be
zero. Hence, $\R(g)=(\mathbf{1}-\PP_{1j'})^{\top}r$. Note that
different polyhedrons similar to $V_2$ might exist with different
$j'$.

Now we consider the regions where~\eqref{eq:vnot} holds for multiple
$j \in \{1,\ldots,n-1\}$.  Let us denote by $V_3$ the polyhedron,
whose few faces are given by
\begin{align*}
 \overline{\omega}_j + \PP_{2j}^{\top} l-\PP_{1j}^{\top} g =\PP_{1j}^{\top} r,
\end{align*}
 for all  $j$ satisfying~\eqref{eq:vnot}. Inside $V_3$,
$\R(g)=\mathbf{1}^{\top} \Delta g$, where $\Delta g$ is given by the
simultaneous solution of 
\begin{align*}
\PP_{1j}^{\top} \Delta g \leq \overline{\omega}_j
+ \PP_{2j}^{\top} - \PP_{1j}^{\top} g,
\end{align*}
for all the corresponding $j$ and $\mathbf{1}^{\top} \Delta g$ is
maximum. At least, one of these inequalities would hold with
equality. Similar to $V_2$, we notice that if we increase some
component(s) of $g$ in $V_3$ with corresponding entry in any of
$P_{ij}$ as 1, $\R(g)$ decreases linearly. While increasing some
component of $g$, a point would be reached where
\begin{align}\label{eq:rc_zero}
\overline{\omega}_j +
\PP_{2j}^{\top} l-\PP_{1j}^{\top} g =0,
\end{align}
for some $j$ and that would be another face of $V_3$. On this
hyperplane, $\Delta g_p=0$ for the controllable nodes for which the
corresponding entry of $\PP_{1j}=1$ in~\eqref{eq:rc_zero}.  Note that
$\R(g)$ is still linear as $V_3$ but with a different slope.

In general, depending on the parameters of the microgrid at hand,
there would be several polyhedrons where~\eqref{eq:vnot} holds
for different $j$. But the characterization of ramping capacity
would be similar to $V_3$ in all these.
Since the ramp rate is either affine or constant in all the
polyhedra, it is affine. 
\hfill $\blacksquare$

\subsection*{Proof of Lemma \ref{lemma:cvx}}
If the difference between two regulation powers, i.e., $|x-x^-|$ is
greater than the ramp rate at $x^-$, then the microgrid might not be
able to provide the regulation power at all. On the other hand, if the
difference
is less than the ramp rate, then it is clear that the microgrid would
be able to provide the required regulation power optimally. So, in the
latter case, the cost of providing regulation power $x$ or the
solution of~\eqref{eq:costnewno} is equivalent to the optimization
in~\eqref{eq:cost}.

Next, we provide a proof for the convexity of $f$ if $h$ is convex. 
Let $C(x)=C_0 \cap C_1(x)$, where $C_0$ denotes the capacity
constraints for $g$ and
\begin{align*}
  C_1(x)= \Big\{g \; | \; \begin{bmatrix} P^0+x
    \\
    g
    \\
    -l
  \end{bmatrix}
  = \M\omega \text{ and } |\omega| \leq \overline{\omega} \Big\}.
\end{align*}
Then, we have $ f(x)= \min\limits_{g \in C(x)} h(g)$.  Let $x_1,x_2
\in [\overline{x},\underline{x}]$, where $\overline{x}$ and
$\underline{x}$ are respectively, the maximum up and down regulation
identified in Section~\ref{sec:bounds}.  Then $f(x_1)=\min\limits_{g
  \in C(x_1)} h(g) $, which means that for all $\delta > 0$, there
exists $g_1 \in C(x_1)$ such that $f(x_1)+ \delta \geq h(g_1)$.
Similarly, there exists $g_2 \in C(x_2)$ such that $f(x_2)+ \delta
\geq h(g_2) $.  Since $g_1 \in C(x_1)$ and $ g_2 \in C(x_2)$,
therefore $ \lambda g_1 + (1-\lambda) g_2 \in C(\lambda x_1 +
(1-\lambda) x_2)$, where $\lambda \in [0,1]$.  Hence,
\begin{align*}
  f(\lambda x_1 + (1-\lambda) x_2) &= \min\limits_{g \in C(\lambda x_1
    + (1-\lambda) x_2)} h(g),
  \\
  & \leq h( \lambda g_1 + (1-\lambda) g_2),\\
  & \leq \lambda h(g_1) + (1-\lambda)h(g_2),\\
  & \leq \lambda f(x_1)+(1-\lambda)f(x_2) + \delta,
\end{align*}
where the second last inequality would be strict in case of strict
convexity.  Since $\delta$ is arbitrary,  $f$ is (strictly)
convex. \hfill $\blacksquare$

\subsection*{Proof of Lemma~\ref{lemma:equivalence}}
We begin by noting that $x_{\im}=x_{\im}^-$ for each $\im$ satisfies
both set of constraints in~\eqref{eq:isonew}, since $x^-$ is the set
of regulations provided by the aggregators at the previous instant.
Hence,~\eqref{eq:isonew} is always feasible.  To prove the equivalence
between the two problems, as our first step, we rewrite~\eqref{eq:iso}
as
\begin{equation}\label{eq:iso1}
  \begin{aligned}
    & {\min_{x}}
    & &   f(x) \\
    & \text{s.t.}  & & x_r \leq \one^{\top} x,
    \\
    &&& \underline{x_{\im}} \leq x_{\im} \leq
    \overline{x_{\im}} \quad \forall \im,
    \\
    &&& |x_{\im} - x_{\im}^-| \leq R_{\im}(x_{\im}^-) \quad \forall \im.
  \end{aligned}
\end{equation}
Note that the equality constraint in~\eqref{eq:iso} is replaced by the
inequality constraint in \eqref{eq:iso1}.  If feasible, both problems
have the same set of solutions. Problem~\eqref{eq:iso1} can still be
infeasible.  Let $\F$ denote its feasible set.  Since $\F$ is compact,
the solution set of~\eqref{eq:iso1} is also compact.  Also, since the
constraints in~\eqref{eq:iso1} are affine, the refined Slater
condition is satisfied.  According to~\cite[Proposition 1]{DPB:75b},
if~\eqref{eq:iso1} is convex, has a non-empty and compact solution set
and satisfies the refined Slater condition, then~\eqref{eq:iso1}
and~\eqref{eq:isonew} have exactly the same solution set if
\begin{align*}
  \mu>\|\lambda \|_\infty,
\end{align*}
for some Lagrange multiplier $\lambda$ of~\eqref{eq:iso1}, as
claimed. \hfill $\blacksquare$

\subsection*{Proof of Theorem \ref{thm:gdac}}
For simplicity of exposition, we ignore the box constraints and write~\eqref{eq:gd+dac} as
\begin{subequations}\label{eq:gd+dac_local}
  \begin{align}
   \!\! \dot{x} & = -\nabla f (x) + [\mu]_{z}^+ \label{eq:x_local},
    \\
   \!\! \dot{z}& = -\nu z -\beta \Lap z - v + \nu (x_r e - x) + \nabla f (x) -
    [\mu]^+_z,\label{eq:z_local}
    \\
    \!\!\dot{v}&=\nu \beta \Lap z, \label{eq:v_local}
  \end{align}
\end{subequations}
First, consider the function $V_2:\real^{2N} \to \real_{\ge 0}$,
$V_2(x,z)=\one^{\top} z - \Delta x$. The Lie
derivative $\lie_{\gdac} V_2 : \real^{2N} \rightrightarrows \real$ is then given by
\begin{align*} 
 \lie_{\gdac} V_2 = \one^{\top} \dot{z} + \one^{\top} \dot{x} = -\nu
  \one^{\top}(z-(x_r e - x))=-\nu V_2,
\end{align*}
where we have used the fact that $\one^{\top} v=0$ due to the initial
condition $\one^{\top} v(0)=0$ and dynamics $\eqref{eq:v_local}$.  The
above equation implies that the summation of all the entries of $z$
converges to the mismatch between the required regulation and procured
regulation exponentially with rate $\nu$. Hence $\one^\top z - \Delta
x \equiv 0$ with the stated initialization.

Next consider the change of coordinates $(x, z,v) \mapsto (x, z ,
\eta)$, with $\eta = \nu (z -(x_r e - x)) + v$. The dynamics for $z$
and $\eta$ are then given by
  \begin{align*}
    \dot{z} &= -\beta \Lap z -\eta + \nabla f(x) - [\mu]^+_z,\\
    \dot{\eta}&=- \nu \eta.
  \end{align*}
Consider the Lyapunov function candidate $V : \real^{3N} \rightarrow \real_{\ge 0}$,
\begin{align*}
  V(x,z,\eta)= f^\mu(x) + \mu \sum\limits_{\im=1}^N [z_{\im}]^+ + \frac{1}{2} \|\eta\|^2,
\end{align*}
whose generalized gradient $\partial V : \real^{3N} \rightrightarrows \real^{3N}$ is given by
\begin{align*}
  \partial V(x,z,\eta)
 \!\! =\!\!
  \begin{cases}
   \! \{ \nabla f(x)\! -\! [ \mu \one]_{\Delta x}^+, [ \mu]_z^+ , \eta \},
    \quad \; \Delta x \neq 0, z \neq \zero,
    \\
   \! \{\nabla f(x)\!-\![ \zero, \mu \one], [ \mu]_z^+, \eta \}, \quad 
   \; \Delta x=0, z \neq \zero,
    \\
    \! \{ \nabla f(x)\! -\! [ \mu \one]_{\Delta x}^+, [\zero ,\mu \one ],
    \eta \},  \Delta x \neq 0, z = \zero,
    \\
   \! \{\nabla f(x)\!-\![ \zero, \mu \one], [\zero, \mu \one], \eta \}, 
   \; \Delta x=0, z = \zero.
  \end{cases}
\end{align*}
Following~\cite{JC:08-csm}, set-valued Lie derivative $\lie_{\gdac}V :
\real^{3N} \rightrightarrows \real$ can then be computed as
\begin{align*}
  \lie_{\gdac}V(x,z,\eta)=
  \begin{cases}
    (\nabla f -[\mu \one]_{\Delta x}^+)^\top (-\nabla f + [ \mu ]_z^+
    )
    \\
    + ([ \mu ]_z^+)^{\top} ( -\beta \Lap z - \eta + \nabla f - [ \mu ]_z^+
    )
    \\
    - \nu \|\eta\|^2, \qquad \qquad \qquad \Delta x \neq 0, z \neq
    \zero, \vspace*{0.2cm}
    \\
    \phi, \qquad \qquad \qquad \qquad \quad \text{otherwise}.
  \end{cases}
\end{align*}
We now analyze various cases of $\Delta x \neq 0, z \neq \zero$ in the
following
\begin{description}
\item[Case 1:] $\Delta x < 0$ and $z < \zero$.
  \begin{align*}
    \lie_{\gdac} V= -\|\nabla f\|^2- \nu \|\eta\|^2.
  \end{align*}
  
\item[Case 2:] $\Delta x > 0$ and $z > \zero$.
  \begin{align*}
  \hspace*{-15pt}  \lie_{\gdac} V= -\|\nabla f\|^2+ 3\mu \nabla f^{\top} \one -
    2N\mu^2 - \mu \eta^{\top} \one - \nu \|\eta\|^2.
  \end{align*}

\item[Case 3:] $\Delta x > 0$ and $z \ngtr \zero$.
  \begin{align*}
    \lie_{\gdac} V =& -\|\nabla f\|^2 - 2N_p\mu^2 + \nabla f^\top (\mu
    \one + 2[\mu]_z^+)\\& -\beta([ \mu]_z^+ )^{\top} \Lap z - \eta^{\top}
    [\mu]_z^+ - \nu \|\eta\|^2,
  \end{align*}
  where $N_p$ is the number of positive elements of $z$.
  
\item[Case 4:] $\Delta x < 0$ and $z \nless \zero$.
  \begin{align*}
    \lie_{\gdac} V =& -\|\nabla f\|^2+ 2 \nabla f^{\top} [ \mu]_z^+ -
    \beta ( [ \mu]_z^+ )^{\top} \Lap z \\ & -\eta^{\top} [ \mu]_z^+ -
    N_p\mu^2 - \nu \|\eta\|^2.
  \end{align*}
\end{description}
We do not need to consider the case when $\Delta x > 0$ and $z <
\zero$ since $\one^\top z -\Delta x \equiv 0$ due to the discussion
above.  Out of the 4 cases, $\lie_{\gdac}V<0$ for Case 1. For the
remaining cases, since $f$ is globally proper and $\| \nabla f\|$ is
bounded over any compact set, $\lie_{\gdac}V < 0$ if the value of
$\mu$ is taken large enough for the worst-case scenario
($N_p=1$). Since $\max \phi=-\infty$, $\max \lie_{\gdac}V<0$ except at
the equilibrium.
This along with the fact that $V$ is locally Lipschitz and regular implies that 
$V$ satisfies the hypothesis of~\cite[Theorem 1]{JC:08-csm}. Hence,
 the dynamics $\gdac$ converge to
the optimal solution asymptotically. \hfill $\blacksquare$

\bibliographystyle{IEEEtran}
\bibliography{alias,Main,Main-add,JC}

\end{document}